\documentclass[runningheads]{llncs}
\usepackage{graphicx} 
\usepackage{tikz}
\usepackage{zed-csp}
\usepackage[T1]{fontenc}
\usepackage{hyperref}

\usepackage[colorinlistoftodos,prependcaption]{todonotes}
\usepackage{enumitem}
\usepackage{cancel,soul,ulem}
\usepackage[linesnumbered,ruled,vlined]{algorithm2e}
\def\tsc#1{\csdef{#1}{\textsc{\lowercase{#1}}\xspace}}
\tsc{WGM}
\tsc{QE}

\title{A Simple Trace Semantics for Asynchronous Sequence Diagrams}

\author{David Faitelson\inst{1} \and
Shmuel Tyszberowicz \inst{1,2}}
\authorrunning{Faitelson and Tyszberowicz}
\institute{Afeka Academic College of Engineering, Israel \\
\email{davidf@afeka.ac.il} 
\and
Centre for Research and Innovation in Software Engineering (RISE), 
School of Computer and Information Science, Southwest University,
Chongqing, China\\
\email{tyshbe@afeka.ac.il}}

\date{May 2024}

\begin{document}
\maketitle

\begin{abstract}
Sequence diagrams are a popular technique for describing interactions between software entities. However, because the OMG group's UML standard is not based on a rigorous mathematical structure, it is impossible to deduce a single interpretation for the notation's semantics, nor to understand precisely how its different fragments interact. While there are a lot of suggested semantics in the literature, they are too mathematically demanding for the majority of software engineers, and often incomplete, especially in dealing with the semantics of lifeline creation and deletion. In this work we describe a simple semantics based on the theory of regular languages, a mathematical theory that is a standard part of the curriculum in every computer science undergraduate degree and covers all the major compositional fragments, and the creation and deletion of lifelines. 
\end{abstract}



\section{Introduction}

Sequence diagrams are a popular technique for describing interactions between software entities. 
Each entity is modeled as a vertical line, and messages between entities are modeled as horizontal arrows. 
The OMG group's UML standard defines a rich syntax for sequence diagrams and describes their semantics informally in terms of traces~\cite{UML17}. 
However, because the standard's definitions are not based on any rigorous mathematical structure, it is impossible to deduce a single interpretation for the notation's semantics, nor to understand precisely how its different fragments interact~\cite{MicskeiW11,MouakherDA22}. 

Even though there is a lot of work on the semantics of sequence diagrams, most of it is quite mathematically sophisticated or requires knowledge and mastery of additional formal notations. Therefore we believe that there is still room for important contributions, in particular by providing an accessible semantics that could be used by practitioners to clarify the meaning of sequence diagrams. 

We have discovered that an important step towards simplifying the semantics is to separate sequence diagrams into two kinds, determined by how we interpret the meaning of lifelines. 
The first kind, synchronous diagrams, represents the interactions between passive objects whose methods are executed by an (implicit) set of threads. 
The other kind, asynchronous diagrams, represents the interaction between active processes (or threads of control). Deciding on the nature of the entities that the lifelines represent is crucial for such diagrams to make sense. 
In particular, synchronous interactions are suitable for passive object diagrams, and asynchronous interactions for active processes. Each kind of diagram may be given simple and consistent semantics. 

However, mixing synchronous and asynchronous interactions in the same diagram is very problematic, as it becomes extremely difficult to intuitively understand what such diagrams mean and to develop a consistent semantics for their behavior. Indeed, most of the research concerning sequence diagram semantics that we are aware of assumes an asynchronous model.  

The purpose of this paper is to define a simple formal denotational semantics for asynchronous diagrams, that covers a large subset of the notation, including nested interaction fragments, as well as the creation and deletion of lifelines. 

In addition, by considering synchronous interactions as restrictions on the pattern of asynchronous interaction, we lay the groundwork for developing the semantics of synchronous interactions in terms of the same underlying semantics. 
However, in this paper, we are concerned only with asynchronous diagrams. 

A common motivation for most works that formalize the semantics of sequence diagrams, is to enable tool support and formal verification. See for example~\cite{Aredo02}.
However, while we certainly agree with this motivation, we believe that it is equally important, and perhaps even more so, to provide a semantics that helps to clarify the meaning of such diagrams, even when we are not using formal methods. 
As the vast majority of engineers are not using formal methods, rigorously clarifying the meaning of sequence diagrams may have a larger positive impact than focusing on formal verification.

The rest of the paper is as follows. In section~\ref{denotational} we introduce an abstract syntax notation for sequence diagrams and define their denotational semantics in terms of sets of traces of atomic messages. 
In the following sections (Sections~\ref{weak} to~\ref{loop}) we describe the specific denotational semantics of each major fragment of the notation. 
This is followed by a section that defines the semantics of lifeline creation and deletion (Section~\ref{create}), and a short section discussing the proper way to describe messages to self (Section~\ref{self}). 
We conclude with a section on related work (Section~\ref{related}), a discussion (Section~\ref{discusssion}), and directions for future work (Section~\ref{summary}). 
To make the ideas easier to follow, we have moved into an appendix all the mathematical development of the functions we use to define the semantics.

\section{Basic definitions}
\label{denotational}

A basic sequence diagram consists of several lifelines with message arrows depicting communications between them. 
Interaction fragments extend the diagram's expressive power by providing a syntax to describe iteration, choice, and interleaving. 
Fragments may nest to create complex descriptions. Therefore, their semantics must be compositional. 

It is difficult to define compositional semantics directly in terms of diagrams because there is no obvious way to identify the separate parts of the diagram from its geometrical representation. 
A better approach\footnote{Several popular tools for drawing sequence diagrams have adopted this approach; see, e.g., \url{https://plantuml.com} or \url{https://www.sequencediagram.org}.} is to create an abstract syntax for the notation and to define both the geometrical representation and the denotational semantics in terms of this abstract syntax. 

Accordingly, we define an abstract syntax for sequence diagram as either a basic sequence of messages exchange or as a composition of interactions using one of the combined fragment operators: $loop$, $alt$~\footnote{The $opt$ operator is a special case of $alt$ with a single child.}, and $par$ (see Fig.~\ref{fig:ast}). 
We focus on these operators because they are the most commonly used in practice. Following the UML, we will call the nodes in the abstract syntax tree \emph{interaction fragments} or $IF$ in short.

\begin{figure}
    \centering
    \begin{syntax}
    IF & ::= & basic \ldata message + \rdata | loop \ldata IF  + \rdata  | alt \ldata IF + \rdata | par \ldata IF + \rdata
    \end{syntax}
    \caption{Abstract syntax for asynchronous sequence diagrams. The $+$ sign indicates a sequence of one or more entities. }
    \label{fig:ast}
\end{figure}

In our semantics, each lifeline represents an independent process or thread of control that interacts with other lifelines by sending and receiving messages. 
The interaction is instantaneous, after which both the sender and the receiver continue without waiting for each other. 

The meaning of an interaction fragment $f$ is a pair $(D(f),N(f))$, consisting of a set of traces $D(f)$ and a set of lifeline names $N(f)$. 
A trace describes one possible sequence of messages exchanged between lifelines. The set of lifeline names provides the scope, that is, the names of the lifelines in the traces of an interaction fragment $f$ must be included in $N(f)$.

A message is a tuple with three components: $(s,l,r)$. 
Here $s$ is the name of the lifeline that sends the message, $l$ is the message's label, and $r$ is the name of the lifeline that receives the message. 

In the following sections, we will define the semantics of each sequence diagram construct. 

\section{Weak sequential composition}
\label{weak}

The messages in a basic interaction fragment are ordered along the lifelines from top to bottom (in the AST representation, a higher message appears to the left of a lower message). 
If two messages are incident on the same lifeline, then the message that appears higher in the diagram occurs before the message that appears lower in the diagram. 
However, the order between messages that do not share the same lifeline is unspecified; this is recorded in the semantics by writing all possible permutations that satisfy the lifeline constraints. 
For example, in each trace of the diagram in Fig.~\ref{fig:asyncweak} message $m_2$ must occur after $m_1$, and message $m_4$ must occur after $m_2$ and $m_3$. However, $m_3$ may occur either before or after $m1$ or $m2$. 
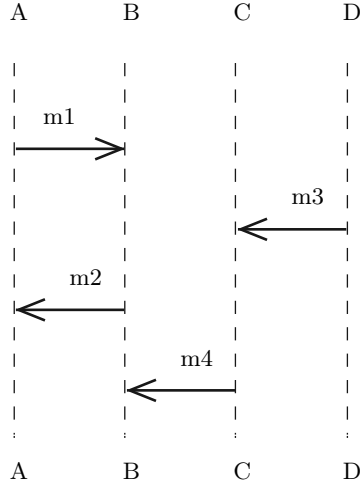
\begin{figure}
    \centering
\definecolor{plantucolor0000}{RGB}{24,24,24}
\definecolor{plantucolor0001}{RGB}{255,255,255}
\definecolor{plantucolor0002}{RGB}{0,0,0}
\begin{tikzpicture}[yscale=-1
,pstyle0/.style={color=plantucolor0000,line width=0.5pt,dash pattern=on 5.0pt off 5.0pt}
,pstyle1/.style={color=plantucolor0001,line width=0.5pt}
,pstyle2/.style={color=plantucolor0000,line width=1.0pt}
]
\draw[pstyle0] (17pt,37.7461pt) -- (17pt,179.6602pt);
\draw[pstyle0] (58.8222pt,37.7461pt) -- (58.8222pt,179.6602pt);
\draw[pstyle0] (100.6667pt,37.7461pt) -- (100.6667pt,179.6602pt);
\draw[pstyle0] (143.0222pt,37.7461pt) -- (143.0222pt,179.6602pt);
\draw[pstyle1] (5pt,10pt) arc (180:270:5pt) -- (10pt,5pt) -- (25.2444pt,5pt) arc (270:360:5pt) -- (30.2444pt,10pt) -- (30.2444pt,31.7461pt) arc (0:90:5pt) -- (25.2444pt,36.7461pt) -- (10pt,36.7461pt) arc (90:180:5pt) -- (5pt,31.7461pt) -- cycle;
\node at (12pt,12pt)[below right,color=black]{A};
\draw[pstyle1] (5pt,183.6602pt) arc (180:270:5pt) -- (10pt,178.6602pt) -- (25.2444pt,178.6602pt) arc (270:360:5pt) -- (30.2444pt,183.6602pt) -- (30.2444pt,205.4063pt) arc (0:90:5pt) -- (25.2444pt,210.4063pt) -- (10pt,210.4063pt) arc (90:180:5pt) -- (5pt,205.4063pt) -- cycle;
\node at (12pt,185.6602pt)[below right,color=black]{A};
\draw[pstyle1] (47.8222pt,10pt) arc (180:270:5pt) -- (52.8222pt,5pt) -- (66.4222pt,5pt) arc (270:360:5pt) -- (71.4222pt,10pt) -- (71.4222pt,31.7461pt) arc (0:90:5pt) -- (66.4222pt,36.7461pt) -- (52.8222pt,36.7461pt) arc (90:180:5pt) -- (47.8222pt,31.7461pt) -- cycle;
\node at (54.8222pt,12pt)[below right,color=black]{B};
\draw[pstyle1] (47.8222pt,183.6602pt) arc (180:270:5pt) -- (52.8222pt,178.6602pt) -- (66.4222pt,178.6602pt) arc (270:360:5pt) -- (71.4222pt,183.6602pt) -- (71.4222pt,205.4063pt) arc (0:90:5pt) -- (66.4222pt,210.4063pt) -- (52.8222pt,210.4063pt) arc (90:180:5pt) -- (47.8222pt,205.4063pt) -- cycle;
\node at (54.8222pt,185.6602pt)[below right,color=black]{B};
\draw[pstyle1] (89.6667pt,10pt) arc (180:270:5pt) -- (94.6667pt,5pt) -- (108.5778pt,5pt) arc (270:360:5pt) -- (113.5778pt,10pt) -- (113.5778pt,31.7461pt) arc (0:90:5pt) -- (108.5778pt,36.7461pt) -- (94.6667pt,36.7461pt) arc (90:180:5pt) -- (89.6667pt,31.7461pt) -- cycle;
\node at (96.6667pt,12pt)[below right,color=black]{C};
\draw[pstyle1] (89.6667pt,183.6602pt) arc (180:270:5pt) -- (94.6667pt,178.6602pt) -- (108.5778pt,178.6602pt) arc (270:360:5pt) -- (113.5778pt,183.6602pt) -- (113.5778pt,205.4063pt) arc (0:90:5pt) -- (108.5778pt,210.4063pt) -- (94.6667pt,210.4063pt) arc (90:180:5pt) -- (89.6667pt,205.4063pt) -- cycle;
\node at (96.6667pt,185.6602pt)[below right,color=black]{C};
\draw[pstyle1] (131.0222pt,10pt) arc (180:270:5pt) -- (136.0222pt,5pt) -- (151.2222pt,5pt) arc (270:360:5pt) -- (156.2222pt,10pt) -- (156.2222pt,31.7461pt) arc (0:90:5pt) -- (151.2222pt,36.7461pt) -- (136.0222pt,36.7461pt) arc (90:180:5pt) -- (131.0222pt,31.7461pt) -- cycle;
\node at (138.0222pt,12pt)[below right,color=black]{D};
\draw[pstyle1] (131.0222pt,183.6602pt) arc (180:270:5pt) -- (136.0222pt,178.6602pt) -- (151.2222pt,178.6602pt) arc (270:360:5pt) -- (156.2222pt,183.6602pt) -- (156.2222pt,205.4063pt) arc (0:90:5pt) -- (151.2222pt,210.4063pt) -- (136.0222pt,210.4063pt) arc (90:180:5pt) -- (131.0222pt,205.4063pt) -- cycle;
\node at (138.0222pt,185.6602pt)[below right,color=black]{D};
\draw[pstyle2] (57.6222pt,70.2246pt) -- (47.6222pt,66.2246pt);
\draw[pstyle2] (57.6222pt,70.2246pt) -- (47.6222pt,74.2246pt);
\draw[pstyle2] (17.6222pt,70.2246pt) -- (58.6222pt,70.2246pt);
\node at (24.6222pt,51.7461pt)[below right,color=black]{m1};
\draw[pstyle2] (102.6222pt,100.7031pt) -- (112.6222pt,96.7031pt);
\draw[pstyle2] (102.6222pt,100.7031pt) -- (112.6222pt,104.7031pt);
\draw[pstyle2] (101.6222pt,100.7031pt) -- (142.6222pt,100.7031pt);
\node at (118.6222pt,82.2246pt)[below right,color=black]{m3};
\draw[pstyle2] (18.6222pt,131.1816pt) -- (28.6222pt,127.1816pt);
\draw[pstyle2] (18.6222pt,131.1816pt) -- (28.6222pt,135.1816pt);
\draw[pstyle2] (17.6222pt,131.1816pt) -- (58.6222pt,131.1816pt);
\node at (34.6222pt,112.7031pt)[below right,color=black]{m2};
\draw[pstyle2] (60.6222pt,161.6602pt) -- (70.6222pt,157.6602pt);
\draw[pstyle2] (60.6222pt,161.6602pt) -- (70.6222pt,165.6602pt);
\draw[pstyle2] (59.6222pt,161.6602pt) -- (100.6222pt,161.6602pt);
\node at (76.6222pt,143.1816pt)[below right,color=black]{m4};
\end{tikzpicture}
    \caption{A simple asynchronous sequence diagram. It represents three traces, according to the possible positions of $m_3$  (either before $m_1$, between $m_1$ and $m_2$, or after $m_2$).}
    \label{fig:asyncweak}
\end{figure}

The abstract syntax of a basic diagram is a list of messages $basic(m_1,\ldots,m_n)$. 
Its semantics is given by the function $weak$ (defined in the appendix). It keeps the sequential order between messages that share a common lifeline, but interleaves messages that do not. 

\begin{zed}
    D(basic(m_1,\ldots, m_n)) = weak~\trace{m_1,\ldots,m_n}
\end{zed}

\noindent
For example, applying $weak$ to the diagram in Fig.~\ref{fig:asyncweak} gives
\begin{zed}
D(basic(m_1,m_3,m_2,m_4)) = 
weak~\trace{m_1,m_3,m_2,m_4} = \\
m_1 \cat weak~\trace{m_3,m_2,m_4} \cup m_3 \cat weak~\trace{m_1,m_2,m_4} = \\
m_1 \cat ( m_3\cat weak~\trace{m_2,m_4} \cup m_2\cat weak~\trace{m_3,m_4}) \cup m_1\cat weak~\trace{m_2,m_4} = \\
m_1\cat(m_3\cat m_2\cat m_4 \cup m_2\cup m_3\cup m_4) \cup m_3\cat m_1\cat m_2\cat m_4 = \\
\{ \trace{m_1,m_3,m_2,m_4}, \trace{m_1,m_2,m_3,m_4}, \trace{m_3,m_1,m_2,m_4} \}
\end{zed}

In the general case, the semantics of weak sequential composition for any two interaction fragments $f1$ and $f2$ is
\begin{zed}
    D(weakseq(f_1,f_2)) = \bigcup\{ t : D(f_1) \cat D(f_2) @ weak~t \}
\end{zed}
We are concatenating the two sets of traces in the usual sense of concatenating each trace from the first set with each trace in the second set.

\section{Alternatives}
\label{alt}

The fragment $alt$ describes situations where the flow could change according to a condition. 
An $alt$ fragment has several branches separated horizontally by dashed lines. 
Each branch may be annotated with a condition expression (see Fig.~\ref{fig:alt}). 
The traces of an $alt$ fragment depend on the conditions. 
When the conditions are mutually exclusive, the traces are taken from the branch whose condition is true. When several conditions may be true at the same time, the choice is nondeterministic 
and the traces are the union of the matching branches. When no conditions are true, the set of traces is empty. 

\begin{figure}
    \centering
\definecolor{plantucolor0000}{RGB}{0,0,0}
\definecolor{plantucolor0001}{RGB}{24,24,24}
\definecolor{plantucolor0002}{RGB}{255,255,255}
\definecolor{plantucolor0003}{RGB}{238,238,238}
\begin{tikzpicture}[yscale=-1
,pstyle0/.style={color=black,line width=1.5pt}
,pstyle1/.style={color=plantucolor0001,line width=0.5pt,dash pattern=on 5.0pt off 5.0pt}
,pstyle2/.style={color=plantucolor0002,line width=0.5pt}
,pstyle4/.style={color=plantucolor0001,line width=1.0pt}
]
\draw[pstyle0] (10pt,54.7461pt) rectangle (107.4pt,150.125pt);
\draw[pstyle1] (32pt,37.7461pt) -- (32pt,197.6035pt);
\draw[pstyle1] (73.8222pt,37.7461pt) -- (73.8222pt,197.6035pt);
\draw[pstyle1] (117.4222pt,37.7461pt) -- (117.4222pt,197.6035pt);
\draw[pstyle2] (20pt,10pt) arc (180:270:5pt) -- (25pt,5pt) -- (40.2444pt,5pt) arc (270:360:5pt) -- (45.2444pt,10pt) -- (45.2444pt,31.7461pt) arc (0:90:5pt) -- (40.2444pt,36.7461pt) -- (25pt,36.7461pt) arc (90:180:5pt) -- (20pt,31.7461pt) -- cycle;
\node at (27pt,12pt)[below right,color=black]{A};
\draw[pstyle2] (20pt,201.6035pt) arc (180:270:5pt) -- (25pt,196.6035pt) -- (40.2444pt,196.6035pt) arc (270:360:5pt) -- (45.2444pt,201.6035pt) -- (45.2444pt,223.3496pt) arc (0:90:5pt) -- (40.2444pt,228.3496pt) -- (25pt,228.3496pt) arc (90:180:5pt) -- (20pt,223.3496pt) -- cycle;
\node at (27pt,203.6035pt)[below right,color=black]{A};
\draw[pstyle2] (62.8222pt,10pt) arc (180:270:5pt) -- (67.8222pt,5pt) -- (81.4222pt,5pt) arc (270:360:5pt) -- (86.4222pt,10pt) -- (86.4222pt,31.7461pt) arc (0:90:5pt) -- (81.4222pt,36.7461pt) -- (67.8222pt,36.7461pt) arc (90:180:5pt) -- (62.8222pt,31.7461pt) -- cycle;
\node at (69.8222pt,12pt)[below right,color=black]{B};
\draw[pstyle2] (62.8222pt,201.6035pt) arc (180:270:5pt) -- (67.8222pt,196.6035pt) -- (81.4222pt,196.6035pt) arc (270:360:5pt) -- (86.4222pt,201.6035pt) -- (86.4222pt,223.3496pt) arc (0:90:5pt) -- (81.4222pt,228.3496pt) -- (67.8222pt,228.3496pt) arc (90:180:5pt) -- (62.8222pt,223.3496pt) -- cycle;
\node at (69.8222pt,203.6035pt)[below right,color=black]{B};
\draw[pstyle2] (96.4222pt,10pt) arc (180:270:5pt) -- (101.4222pt,5pt) -- (115.3333pt,5pt) arc (270:360:5pt) -- (120.3333pt,10pt) -- (120.3333pt,31.7461pt) arc (0:90:5pt) -- (115.3333pt,36.7461pt) -- (101.4222pt,36.7461pt) arc (90:180:5pt) -- (96.4222pt,31.7461pt) -- cycle;
\node at (113.4222pt,12pt)[below right,color=black]{C};
\draw[pstyle2] (96.4222pt,201.6035pt) arc (180:270:5pt) -- (101.4222pt,196.6035pt) -- (115.3333pt,196.6035pt) arc (270:360:5pt) -- (120.3333pt,201.6035pt) -- (120.3333pt,223.3496pt) arc (0:90:5pt) -- (115.3333pt,228.3496pt) -- (101.4222pt,228.3496pt) arc (90:180:5pt) -- (96.4222pt,223.3496pt) -- cycle;
\node at (113.4222pt,203.6035pt)[below right,color=black]{C};
\draw[color=black,fill=plantucolor0003,line width=1.5pt] (10pt,54.7461pt) -- (71.8pt,54.7461pt) -- (71.8pt,63.2246pt) -- (61.8pt,73.2246pt) -- (10pt,73.2246pt) -- (10pt,54.7461pt);
\draw[pstyle0] (10pt,54.7461pt) rectangle (107.4pt,150.125pt);
\node at (25pt,55.7461pt)[below right,color=black]{\textbf{alt}};
\node at (86.8pt,56.7461pt)[below right,color=black]{\textbf{[P]}};
\draw[pstyle4] (72.6222pt,95.7031pt) -- (62.6222pt,91.7031pt);
\draw[pstyle4] (72.6222pt,95.7031pt) -- (62.6222pt,99.7031pt);
\draw[pstyle4] (32.6222pt,95.7031pt) -- (73.6222pt,95.7031pt);
\node at (39.6222pt,77.2246pt)[below right,color=black]{m1};
\draw[color=black,line width=1.0pt,dash pattern=on 2.0pt off 2.0pt] (10pt,104.7031pt) -- (107.4pt,104.7031pt);
\node at (15pt,104.7031pt)[below right,color=black]{\textbf{[not P]}};
\draw[pstyle4] (72.6222pt,142.125pt) -- (62.6222pt,138.125pt);
\draw[pstyle4] (72.6222pt,142.125pt) -- (62.6222pt,146.125pt);
\draw[pstyle4] (32.6222pt,142.125pt) -- (73.6222pt,142.125pt);
\node at (39.6222pt,123.6465pt)[below right,color=black]{m2};
\draw[pstyle4] (116.3778pt,179.6035pt) -- (106.3778pt,175.6035pt);
\draw[pstyle4] (116.3778pt,179.6035pt) -- (106.3778pt,183.6035pt);
\draw[pstyle4] (32.6222pt,179.6035pt) -- (117.3778pt,179.6035pt);
\node at (39.6222pt,161.125pt)[below right,color=black]{m3};
\end{tikzpicture}
    \caption{An alternative fragment. Process $A$ may send either message $m1$ or $m2$ to process $B$ depending on the value of the boolean expression $P$. Afterwards $A$ sends $m3$ to $C$.}
    \label{fig:alt}
\end{figure}
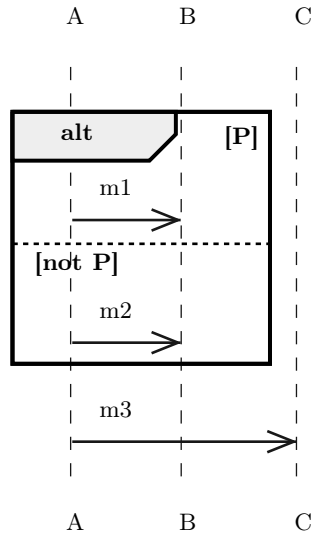

The abstract syntax of an alternative interaction fragment is $alt(f_1,\ldots,f_n)$.
We do not model the conditions because this requires that we model the expressions that appear in them, and thus an entire formal language apparatus for evaluating their values, significantly cluttering the semantics without gaining much value. Therefore, we assume that any alternative may occur, that is, the semantics will include the traces of all alternatives:
\begin{zed}
    D(alt(f_1,\ldots, f_n)) = \bigcup_{i=1}^n D(f_i)
\end{zed}

\section{Parallel interleavings}
\label{interleave}

In some cases, the order of communication is not important. 
A basic diagram, however, insists on a particular order; we could attempt to draw all possible permutations, but this is tedious and may overwhelm the reader. The $par$ fragment overcomes this problem. 
For example, Fig.~\ref{fig:asyncpar} shows a coordinator inviting two participants to a meeting. 
The participants can accept the invitation only after they received it, but the coordinator may invite the participants in any order. 

\begin{figure}
    \centering
\definecolor{plantucolor0000}{RGB}{0,0,0}
\definecolor{plantucolor0001}{RGB}{24,24,24}
\definecolor{plantucolor0002}{RGB}{255,255,255}
\definecolor{plantucolor0003}{RGB}{238,238,238}
\begin{tikzpicture}[yscale=-1
,pstyle0/.style={color=black,line width=1.5pt}
,pstyle1/.style={color=plantucolor0001,line width=0.5pt,dash pattern=on 5.0pt off 5.0pt}
,pstyle2/.style={color=plantucolor0002,line width=0.5pt}
,pstyle4/.style={color=plantucolor0001,line width=1.0pt}
]
\draw[pstyle0] (10pt,54.7461pt) rectangle (339.8286pt,197.1387pt);
\draw[pstyle1] (68pt,37.7461pt) -- (68pt,214.1387pt);
\draw[pstyle1] (174.1429pt,37.7461pt) -- (174.1429pt,214.1387pt);
\draw[pstyle1] (280.8571pt,37.7461pt) -- (280.8571pt,214.1387pt);
\draw[pstyle2] (20pt,10pt) arc (180:270:5pt) -- (25pt,5pt) -- (112.1429pt,5pt) arc (270:360:5pt) -- (117.1429pt,10pt) -- (117.1429pt,31.7461pt) arc (0:90:5pt) -- (112.1429pt,36.7461pt) -- (25pt,36.7461pt) arc (90:180:5pt) -- (20pt,31.7461pt) -- cycle;
\node at (27pt,12pt)[below right,color=black]{\underline{c:Coordinator}};
\draw[pstyle2] (20pt,218.1387pt) arc (180:270:5pt) -- (25pt,213.1387pt) -- (112.1429pt,213.1387pt) arc (270:360:5pt) -- (117.1429pt,218.1387pt) -- (117.1429pt,239.8848pt) arc (0:90:5pt) -- (112.1429pt,244.8848pt) -- (25pt,244.8848pt) arc (90:180:5pt) -- (20pt,239.8848pt) -- cycle;
\node at (27pt,220.1387pt)[below right,color=black]{\underline{c:Coordinator}};
\draw[pstyle2] (127.1429pt,10pt) arc (180:270:5pt) -- (132.1429pt,5pt) -- (217.8571pt,5pt) arc (270:360:5pt) -- (222.8571pt,10pt) -- (222.8571pt,31.7461pt) arc (0:90:5pt) -- (217.8571pt,36.7461pt) -- (132.1429pt,36.7461pt) arc (90:180:5pt) -- (127.1429pt,31.7461pt) -- cycle;
\node at (134.1429pt,12pt)[below right,color=black]{\underline{x:Participant}};
\draw[pstyle2] (127.1429pt,218.1387pt) arc (180:270:5pt) -- (132.1429pt,213.1387pt) -- (217.8571pt,213.1387pt) arc (270:360:5pt) -- (222.8571pt,218.1387pt) -- (222.8571pt,239.8848pt) arc (0:90:5pt) -- (217.8571pt,244.8848pt) -- (132.1429pt,244.8848pt) arc (90:180:5pt) -- (127.1429pt,239.8848pt) -- cycle;
\node at (134.1429pt,220.1387pt)[below right,color=black]{\underline{x:Participant}};
\draw[pstyle2] (232.8571pt,10pt) arc (180:270:5pt) -- (237.8571pt,5pt) -- (324.8286pt,5pt) arc (270:360:5pt) -- (329.8286pt,10pt) -- (329.8286pt,31.7461pt) arc (0:90:5pt) -- (324.8286pt,36.7461pt) -- (237.8571pt,36.7461pt) arc (90:180:5pt) -- (232.8571pt,31.7461pt) -- cycle;
\node at (239.8571pt,12pt)[below right,color=black]{\underline{y:Participant}};
\draw[pstyle2] (232.8571pt,218.1387pt) arc (180:270:5pt) -- (237.8571pt,213.1387pt) -- (324.8286pt,213.1387pt) arc (270:360:5pt) -- (329.8286pt,218.1387pt) -- (329.8286pt,239.8848pt) arc (0:90:5pt) -- (324.8286pt,244.8848pt) -- (237.8571pt,244.8848pt) arc (90:180:5pt) -- (232.8571pt,239.8848pt) -- cycle;
\node at (239.8571pt,220.1387pt)[below right,color=black]{\underline{y:Participant}};
\draw[color=black,fill=plantucolor0003,line width=1.5pt] (10pt,54.7461pt) -- (75.9333pt,54.7461pt) -- (75.9333pt,63.2246pt) -- (65.9333pt,73.2246pt) -- (10pt,73.2246pt) -- (10pt,54.7461pt);
\draw[pstyle0] (10pt,54.7461pt) rectangle (339.8286pt,197.1387pt);
\node at (25pt,55.7461pt)[below right,color=black]{\textbf{par}};
\draw[pstyle4] (173pt,95.7031pt) -- (163pt,91.7031pt);
\draw[pstyle4] (173pt,95.7031pt) -- (163pt,99.7031pt);
\draw[pstyle4] (68.5714pt,95.7031pt) -- (174pt,95.7031pt);
\node at (75.5714pt,77.2246pt)[below right,color=black]{invite};
\draw[pstyle4] (69.5714pt,126.1816pt) -- (79.5714pt,122.1816pt);
\draw[pstyle4] (69.5714pt,126.1816pt) -- (79.5714pt,130.1816pt);
\draw[pstyle4] (68.5714pt,126.1816pt) -- (174pt,126.1816pt);
\node at (85.5714pt,107.7031pt)[below right,color=black]{accept};
\draw[color=black,line width=1.0pt,dash pattern=on 2.0pt off 2.0pt] (10pt,135.1816pt) -- (339.8286pt,135.1816pt);
\draw[pstyle4] (279.3429pt,158.6602pt) -- (269.3429pt,154.6602pt);
\draw[pstyle4] (279.3429pt,158.6602pt) -- (269.3429pt,162.6602pt);
\draw[pstyle4] (68.5714pt,158.6602pt) -- (280.3429pt,158.6602pt);
\node at (75.5714pt,140.1816pt)[below right,color=black]{invite};
\draw[pstyle4] (69.5714pt,189.1387pt) -- (79.5714pt,185.1387pt);
\draw[pstyle4] (69.5714pt,189.1387pt) -- (79.5714pt,193.1387pt);
\draw[pstyle4] (68.5714pt,189.1387pt) -- (280.3429pt,189.1387pt);
\node at (85.5714pt,170.6602pt)[below right,color=black]{accept};
\end{tikzpicture}
    \caption{A parallel combined fragment. The coordinator may invite the participants in any order. However, each participant accepts only after they have been invited. }
    \label{fig:asyncpar}
\end{figure}
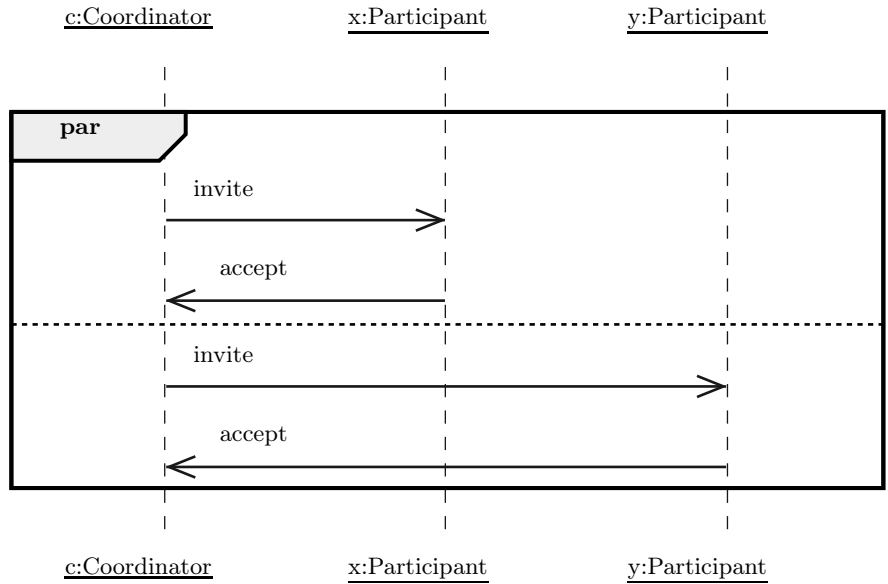

To simplify the discussion, we describe the semantics of a parallel fragment with just two children. 
In a fragment with more children, we apply the interleaving operator repeatedly.

\begin{zed}
    D(par(f_1,f_2)) = interleave~(D(f_1),D(f_2))
\end{zed}

The $interleave$ function calculates the set of possible interleavings of each trace in its first argument with each trace in its second argument. For example, calculating the traces for the diagram in Fig.~\ref{fig:asyncpar} gives
\begin{zed}
    D(par(basic(c.invite.x, x.accept.c), basic(c.invite.y, y.accept.c))) = \\ 
    interleave(D(basic(c.invite.x, x.accept.c)), D(basic(c.invite.y, y.accept.c))) = \\
    interleave(\{\trace{c.invite.x,x.accept.c}\}, \{\trace{c.invite.y,y.accept.c}\}) = \\
    \trace{c.invite.x,x.accept.c} \interleave \trace{c.invite.y,y.accept.c} = \\
    \{ \trace {c.invite.x, x.accept.c, c.invite.y, y.accept.c } , \\
     \trace {c.invite.x, c.invite.y, x.accept.c, y.accept.c }, \\
     \trace {c.invite.x, c.invite.y, y.accept.c, x.accept.c }, \\
     \trace {c.invite.y, c.invite.x, x.accept.c, y.accept.c }, \\
     \trace {c.invite.y, c.invite.x, y.accept.c, x.accept.c},  \\
     \trace {c.invite.y, y.accept.c, c.invite.x, x.accept.c} \}
\end{zed}

\section{Loops}
\label{loop}

Loops are used to describe iterative behaviors. 
To simplify the discussion, we consider loop fragments that do not have a terminating guard. Such loops represent the behavior of an unbounded number of iterations. To explain the loop's semantics it is useful to define several operators on sets of traces. Given sets of traces $u$ and $v$,

\begin{zed}
    u v = \{ x \in u \land y \in v @ x\cat y \} 
\also 
    u^0 = \{ \trace{} \}    
\also 
    u^{n+1} = u^n u
\also
    u^{*} = \cup_{n=0}^\infty u^n
\end{zed}

These are the familiar string and language operators we use in the theory of computation~\cite{rich}, except our strings are traces (and the messages are the symbols). 

As a first approximation, consider a loop $loop(f)$, then we may define the loop's traces to be the Kleene closure of $f$'s traces, that is, as $D(f)^*$. 
For example, the traces of the loop's body in Fig.~\ref{fig:asyncloop} require that $m_2$ appears after $m_1$ and $m_4$ after $m_3$. This constraint holds in each iteration. However, the weak sequential semantics does not imply that all four processes will communicate in each iteration. It is possible for example that the first $k$ iterations will involve just $A$ and $B$ before $C$ and $D$ communicate. This becomes obvious when we unfold the loop and consider the weak semantics of the unfolded version. See Fig.~\ref{fig:unfoldedloop}.

\begin{figure}
    \centering
\definecolor{plantucolor0000}{RGB}{0,0,0}
\definecolor{plantucolor0001}{RGB}{24,24,24}
\definecolor{plantucolor0002}{RGB}{255,255,255}
\definecolor{plantucolor0003}{RGB}{238,238,238}
\begin{tikzpicture}[yscale=-1
,pstyle0/.style={color=black,line width=1.5pt}
,pstyle1/.style={color=plantucolor0001,line width=0.5pt,dash pattern=on 5.0pt off 5.0pt}
,pstyle2/.style={color=plantucolor0002,line width=0.5pt}
,pstyle4/.style={color=plantucolor0001,line width=1.0pt}
]
\draw[pstyle0] (10pt,54.7461pt) rectangle (172.9778pt,195.1387pt);
\draw[pstyle1] (32pt,37.7461pt) -- (32pt,212.1387pt);
\draw[pstyle1] (73.8222pt,37.7461pt) -- (73.8222pt,212.1387pt);
\draw[pstyle1] (107.4222pt,37.7461pt) -- (107.4222pt,212.1387pt);
\draw[pstyle1] (149.7778pt,37.7461pt) -- (149.7778pt,212.1387pt);
\draw[pstyle2] (20pt,10pt) arc (180:270:5pt) -- (25pt,5pt) -- (40.2444pt,5pt) arc (270:360:5pt) -- (45.2444pt,10pt) -- (45.2444pt,31.7461pt) arc (0:90:5pt) -- (40.2444pt,36.7461pt) -- (25pt,36.7461pt) arc (90:180:5pt) -- (20pt,31.7461pt) -- cycle;
\node at (27pt,12pt)[below right,color=black]{A};
\draw[pstyle2] (20pt,216.1387pt) arc (180:270:5pt) -- (25pt,211.1387pt) -- (40.2444pt,211.1387pt) arc (270:360:5pt) -- (45.2444pt,216.1387pt) -- (45.2444pt,237.8848pt) arc (0:90:5pt) -- (40.2444pt,242.8848pt) -- (25pt,242.8848pt) arc (90:180:5pt) -- (20pt,237.8848pt) -- cycle;
\node at (27pt,218.1387pt)[below right,color=black]{A};
\draw[pstyle2] (62.8222pt,10pt) arc (180:270:5pt) -- (67.8222pt,5pt) -- (81.4222pt,5pt) arc (270:360:5pt) -- (86.4222pt,10pt) -- (86.4222pt,31.7461pt) arc (0:90:5pt) -- (81.4222pt,36.7461pt) -- (67.8222pt,36.7461pt) arc (90:180:5pt) -- (62.8222pt,31.7461pt) -- cycle;
\node at (69.8222pt,12pt)[below right,color=black]{B};
\draw[pstyle2] (62.8222pt,216.1387pt) arc (180:270:5pt) -- (67.8222pt,211.1387pt) -- (81.4222pt,211.1387pt) arc (270:360:5pt) -- (86.4222pt,216.1387pt) -- (86.4222pt,237.8848pt) arc (0:90:5pt) -- (81.4222pt,242.8848pt) -- (67.8222pt,242.8848pt) arc (90:180:5pt) -- (62.8222pt,237.8848pt) -- cycle;
\node at (69.8222pt,218.1387pt)[below right,color=black]{B};
\draw[pstyle2] (96.4222pt,10pt) arc (180:270:5pt) -- (101.4222pt,5pt) -- (115.3333pt,5pt) arc (270:360:5pt) -- (120.3333pt,10pt) -- (120.3333pt,31.7461pt) arc (0:90:5pt) -- (115.3333pt,36.7461pt) -- (101.4222pt,36.7461pt) arc (90:180:5pt) -- (96.4222pt,31.7461pt) -- cycle;
\node at (103.4222pt,12pt)[below right,color=black]{C};
\draw[pstyle2] (96.4222pt,216.1387pt) arc (180:270:5pt) -- (101.4222pt,211.1387pt) -- (115.3333pt,211.1387pt) arc (270:360:5pt) -- (120.3333pt,216.1387pt) -- (120.3333pt,237.8848pt) arc (0:90:5pt) -- (115.3333pt,242.8848pt) -- (101.4222pt,242.8848pt) arc (90:180:5pt) -- (96.4222pt,237.8848pt) -- cycle;
\node at (103.4222pt,218.1387pt)[below right,color=black]{C};
\draw[pstyle2] (137.7778pt,10pt) arc (180:270:5pt) -- (142.7778pt,5pt) -- (157.9778pt,5pt) arc (270:360:5pt) -- (162.9778pt,10pt) -- (162.9778pt,31.7461pt) arc (0:90:5pt) -- (157.9778pt,36.7461pt) -- (142.7778pt,36.7461pt) arc (90:180:5pt) -- (137.7778pt,31.7461pt) -- cycle;
\node at (144.7778pt,12pt)[below right,color=black]{D};
\draw[pstyle2] (137.7778pt,216.1387pt) arc (180:270:5pt) -- (142.7778pt,211.1387pt) -- (157.9778pt,211.1387pt) arc (270:360:5pt) -- (162.9778pt,216.1387pt) -- (162.9778pt,237.8848pt) arc (0:90:5pt) -- (157.9778pt,242.8848pt) -- (142.7778pt,242.8848pt) arc (90:180:5pt) -- (137.7778pt,237.8848pt) -- cycle;
\node at (144.7778pt,218.1387pt)[below right,color=black]{D};
\draw[color=black,fill=plantucolor0003,line width=1.5pt] (10pt,54.7461pt) -- (81.4pt,54.7461pt) -- (81.4pt,63.2246pt) -- (71.4pt,73.2246pt) -- (10pt,73.2246pt) -- (10pt,54.7461pt);
\draw[pstyle0] (10pt,54.7461pt) rectangle (172.9778pt,195.1387pt);
\node at (25pt,55.7461pt)[below right,color=black]{\textbf{loop}};
\draw[pstyle4] (72.6222pt,95.7031pt) -- (62.6222pt,91.7031pt);
\draw[pstyle4] (72.6222pt,95.7031pt) -- (62.6222pt,99.7031pt);
\draw[pstyle4] (32.6222pt,95.7031pt) -- (73.6222pt,95.7031pt);
\node at (39.6222pt,77.2246pt)[below right,color=black]{m1};
\draw[pstyle4] (109.3778pt,126.1816pt) -- (119.3778pt,122.1816pt);
\draw[pstyle4] (109.3778pt,126.1816pt) -- (119.3778pt,130.1816pt);
\draw[pstyle4] (108.3778pt,126.1816pt) -- (149.3778pt,126.1816pt);
\node at (125.3778pt,107.7031pt)[below right,color=black]{m3};
\draw[pstyle4] (72.6222pt,156.6602pt) -- (62.6222pt,152.6602pt);
\draw[pstyle4] (72.6222pt,156.6602pt) -- (62.6222pt,160.6602pt);
\draw[pstyle4] (32.6222pt,156.6602pt) -- (73.6222pt,156.6602pt);
\node at (39.6222pt,138.1816pt)[below right,color=black]{m2};
\draw[pstyle4] (109.3778pt,187.1387pt) -- (119.3778pt,183.1387pt);
\draw[pstyle4] (109.3778pt,187.1387pt) -- (119.3778pt,191.1387pt);
\draw[pstyle4] (108.3778pt,187.1387pt) -- (149.3778pt,187.1387pt);
\node at (125.3778pt,168.6602pt)[below right,color=black]{m4};
\end{tikzpicture}
    \caption{Asynchronous loop}
    \label{fig:asyncloop}
\end{figure}
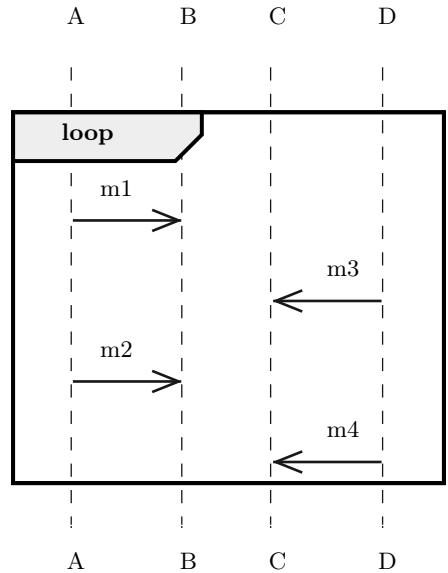

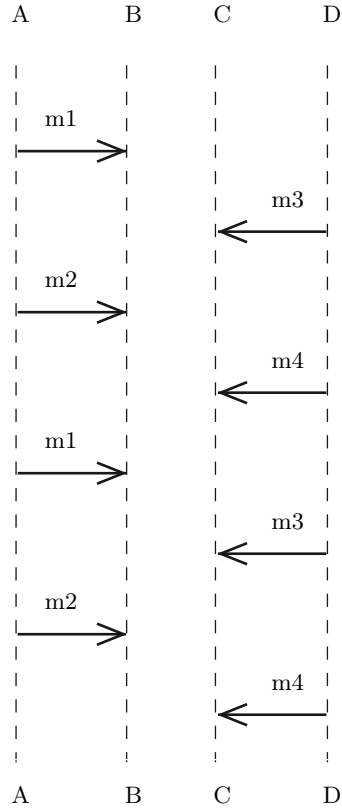
\begin{figure}[tb]
    \centering
\definecolor{plantucolor0000}{RGB}{24,24,24}
\definecolor{plantucolor0001}{RGB}{255,255,255}
\definecolor{plantucolor0002}{RGB}{0,0,0}
\begin{tikzpicture}[yscale=-1
,pstyle0/.style={color=plantucolor0000,line width=0.5pt,dash pattern=on 5.0pt off 5.0pt}
,pstyle1/.style={color=plantucolor0001,line width=0.5pt}
,pstyle2/.style={color=plantucolor0000,line width=1.0pt}
]
\draw[pstyle0] (17pt,37.7461pt) -- (17pt,301.5742pt);
\draw[pstyle0] (58.8222pt,37.7461pt) -- (58.8222pt,301.5742pt);
\draw[pstyle0] (92.4222pt,37.7461pt) -- (92.4222pt,301.5742pt);
\draw[pstyle0] (134.7778pt,37.7461pt) -- (134.7778pt,301.5742pt);
\draw[pstyle1] (5pt,10pt) arc (180:270:5pt) -- (10pt,5pt) -- (25.2444pt,5pt) arc (270:360:5pt) -- (30.2444pt,10pt) -- (30.2444pt,31.7461pt) arc (0:90:5pt) -- (25.2444pt,36.7461pt) -- (10pt,36.7461pt) arc (90:180:5pt) -- (5pt,31.7461pt) -- cycle;
\node at (12pt,12pt)[below right,color=black]{A};
\draw[pstyle1] (5pt,305.5742pt) arc (180:270:5pt) -- (10pt,300.5742pt) -- (25.2444pt,300.5742pt) arc (270:360:5pt) -- (30.2444pt,305.5742pt) -- (30.2444pt,327.3203pt) arc (0:90:5pt) -- (25.2444pt,332.3203pt) -- (10pt,332.3203pt) arc (90:180:5pt) -- (5pt,327.3203pt) -- cycle;
\node at (12pt,307.5742pt)[below right,color=black]{A};
\draw[pstyle1] (47.8222pt,10pt) arc (180:270:5pt) -- (52.8222pt,5pt) -- (66.4222pt,5pt) arc (270:360:5pt) -- (71.4222pt,10pt) -- (71.4222pt,31.7461pt) arc (0:90:5pt) -- (66.4222pt,36.7461pt) -- (52.8222pt,36.7461pt) arc (90:180:5pt) -- (47.8222pt,31.7461pt) -- cycle;
\node at (54.8222pt,12pt)[below right,color=black]{B};
\draw[pstyle1] (47.8222pt,305.5742pt) arc (180:270:5pt) -- (52.8222pt,300.5742pt) -- (66.4222pt,300.5742pt) arc (270:360:5pt) -- (71.4222pt,305.5742pt) -- (71.4222pt,327.3203pt) arc (0:90:5pt) -- (66.4222pt,332.3203pt) -- (52.8222pt,332.3203pt) arc (90:180:5pt) -- (47.8222pt,327.3203pt) -- cycle;
\node at (54.8222pt,307.5742pt)[below right,color=black]{B};
\draw[pstyle1] (81.4222pt,10pt) arc (180:270:5pt) -- (86.4222pt,5pt) -- (100.3333pt,5pt) arc (270:360:5pt) -- (105.3333pt,10pt) -- (105.3333pt,31.7461pt) arc (0:90:5pt) -- (100.3333pt,36.7461pt) -- (86.4222pt,36.7461pt) arc (90:180:5pt) -- (81.4222pt,31.7461pt) -- cycle;
\node at (88.4222pt,12pt)[below right,color=black]{C};
\draw[pstyle1] (81.4222pt,305.5742pt) arc (180:270:5pt) -- (86.4222pt,300.5742pt) -- (100.3333pt,300.5742pt) arc (270:360:5pt) -- (105.3333pt,305.5742pt) -- (105.3333pt,327.3203pt) arc (0:90:5pt) -- (100.3333pt,332.3203pt) -- (86.4222pt,332.3203pt) arc (90:180:5pt) -- (81.4222pt,327.3203pt) -- cycle;
\node at (88.4222pt,307.5742pt)[below right,color=black]{C};
\draw[pstyle1] (122.7778pt,10pt) arc (180:270:5pt) -- (127.7778pt,5pt) -- (142.9778pt,5pt) arc (270:360:5pt) -- (147.9778pt,10pt) -- (147.9778pt,31.7461pt) arc (0:90:5pt) -- (142.9778pt,36.7461pt) -- (127.7778pt,36.7461pt) arc (90:180:5pt) -- (122.7778pt,31.7461pt) -- cycle;
\node at (129.7778pt,12pt)[below right,color=black]{D};
\draw[pstyle1] (122.7778pt,305.5742pt) arc (180:270:5pt) -- (127.7778pt,300.5742pt) -- (142.9778pt,300.5742pt) arc (270:360:5pt) -- (147.9778pt,305.5742pt) -- (147.9778pt,327.3203pt) arc (0:90:5pt) -- (142.9778pt,332.3203pt) -- (127.7778pt,332.3203pt) arc (90:180:5pt) -- (122.7778pt,327.3203pt) -- cycle;
\node at (129.7778pt,307.5742pt)[below right,color=black]{D};
\draw[pstyle2] (57.6222pt,70.2246pt) -- (47.6222pt,66.2246pt);
\draw[pstyle2] (57.6222pt,70.2246pt) -- (47.6222pt,74.2246pt);
\draw[pstyle2] (17.6222pt,70.2246pt) -- (58.6222pt,70.2246pt);
\node at (24.6222pt,51.7461pt)[below right,color=black]{m1};
\draw[pstyle2] (94.3778pt,100.7031pt) -- (104.3778pt,96.7031pt);
\draw[pstyle2] (94.3778pt,100.7031pt) -- (104.3778pt,104.7031pt);
\draw[pstyle2] (93.3778pt,100.7031pt) -- (134.3778pt,100.7031pt);
\node at (110.3778pt,82.2246pt)[below right,color=black]{m3};
\draw[pstyle2] (57.6222pt,131.1816pt) -- (47.6222pt,127.1816pt);
\draw[pstyle2] (57.6222pt,131.1816pt) -- (47.6222pt,135.1816pt);
\draw[pstyle2] (17.6222pt,131.1816pt) -- (58.6222pt,131.1816pt);
\node at (24.6222pt,112.7031pt)[below right,color=black]{m2};
\draw[pstyle2] (94.3778pt,161.6602pt) -- (104.3778pt,157.6602pt);
\draw[pstyle2] (94.3778pt,161.6602pt) -- (104.3778pt,165.6602pt);
\draw[pstyle2] (93.3778pt,161.6602pt) -- (134.3778pt,161.6602pt);
\node at (110.3778pt,143.1816pt)[below right,color=black]{m4};
\draw[pstyle2] (57.6222pt,192.1387pt) -- (47.6222pt,188.1387pt);
\draw[pstyle2] (57.6222pt,192.1387pt) -- (47.6222pt,196.1387pt);
\draw[pstyle2] (17.6222pt,192.1387pt) -- (58.6222pt,192.1387pt);
\node at (24.6222pt,173.6602pt)[below right,color=black]{m1};
\draw[pstyle2] (94.3778pt,222.6172pt) -- (104.3778pt,218.6172pt);
\draw[pstyle2] (94.3778pt,222.6172pt) -- (104.3778pt,226.6172pt);
\draw[pstyle2] (93.3778pt,222.6172pt) -- (134.3778pt,222.6172pt);
\node at (110.3778pt,204.1387pt)[below right,color=black]{m3};
\draw[pstyle2] (57.6222pt,253.0957pt) -- (47.6222pt,249.0957pt);
\draw[pstyle2] (57.6222pt,253.0957pt) -- (47.6222pt,257.0957pt);
\draw[pstyle2] (17.6222pt,253.0957pt) -- (58.6222pt,253.0957pt);
\node at (24.6222pt,234.6172pt)[below right,color=black]{m2};
\draw[pstyle2] (94.3778pt,283.5742pt) -- (104.3778pt,279.5742pt);
\draw[pstyle2] (94.3778pt,283.5742pt) -- (104.3778pt,287.5742pt);
\draw[pstyle2] (93.3778pt,283.5742pt) -- (134.3778pt,283.5742pt);
\node at (110.3778pt,265.0957pt)[below right,color=black]{m4};
\end{tikzpicture}
    \caption{Unfolding twice the loop in Fig.~\ref{fig:asyncloop}. The communication between $A$ and $B$ is interleaved with the communication between $C$ and $D$. In particular, all communication between $C$ and $D$ may occur before any communication between $A$ and $B$ (and vice versa). }
    \label{fig:unfoldedloop}
\end{figure}

Thus, the semantics of a loop fragment is given by applying $weak$ to each trace in the Kleene closure of its body:
\begin{zed}
    D(loop(b)) = \bigcup \{ t : D(b)^* @ weak(t) \}
\end{zed}

\section{Creating and deleting lifelines}
\label{create}

We create a new lifeline by sending a special $create$ message that starts the new lifeline immediately after the message is received. We destroy a lifeline by sending a special $destroy$ message to the lifeline. Destruction is indicated by an X sign, after which the lifeline ends.  

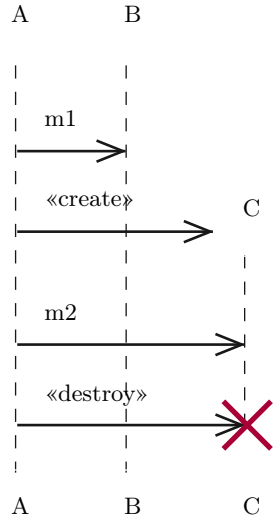
\begin{figure}
    \begin{center}
\definecolor{plantucolor0000}{RGB}{24,24,24}
\definecolor{plantucolor0001}{RGB}{255,255,255}
\definecolor{plantucolor0002}{RGB}{0,0,0}
\definecolor{plantucolor0003}{RGB}{168,0,54}
\begin{tikzpicture}[yscale=-1
,pstyle0/.style={color=plantucolor0000,line width=0.5pt,dash pattern=on 5.0pt off 5.0pt}
,pstyle1/.style={color=plantucolor0001,line width=0.5pt}
,pstyle2/.style={color=plantucolor0000,line width=1.0pt}
,pstyle3/.style={color=plantucolor0003,line width=2.0pt}
]
\draw[pstyle0] (17pt,37.7461pt) -- (17pt,191.9277pt);
\draw[pstyle0] (58.8222pt,37.7461pt) -- (58.8222pt,191.9277pt);
\draw[pstyle0] (103.6483pt,109.5977pt) -- (103.6483pt,173.9277pt);
\draw[pstyle1] (5pt,10pt) arc (180:270:5pt) -- (10pt,5pt) -- (25.2444pt,5pt) arc (270:360:5pt) -- (30.2444pt,10pt) -- (30.2444pt,31.7461pt) arc (0:90:5pt) -- (25.2444pt,36.7461pt) -- (10pt,36.7461pt) arc (90:180:5pt) -- (5pt,31.7461pt) -- cycle;
\node at (12pt,12pt)[below right,color=black]{A};
\draw[pstyle1] (5pt,195.9277pt) arc (180:270:5pt) -- (10pt,190.9277pt) -- (25.2444pt,190.9277pt) arc (270:360:5pt) -- (30.2444pt,195.9277pt) -- (30.2444pt,217.6738pt) arc (0:90:5pt) -- (25.2444pt,222.6738pt) -- (10pt,222.6738pt) arc (90:180:5pt) -- (5pt,217.6738pt) -- cycle;
\node at (12pt,197.9277pt)[below right,color=black]{A};
\draw[pstyle1] (47.8222pt,10pt) arc (180:270:5pt) -- (52.8222pt,5pt) -- (66.4222pt,5pt) arc (270:360:5pt) -- (71.4222pt,10pt) -- (71.4222pt,31.7461pt) arc (0:90:5pt) -- (66.4222pt,36.7461pt) -- (52.8222pt,36.7461pt) arc (90:180:5pt) -- (47.8222pt,31.7461pt) -- cycle;
\node at (54.8222pt,12pt)[below right,color=black]{B};
\draw[pstyle1] (47.8222pt,195.9277pt) arc (180:270:5pt) -- (52.8222pt,190.9277pt) -- (66.4222pt,190.9277pt) arc (270:360:5pt) -- (71.4222pt,195.9277pt) -- (71.4222pt,217.6738pt) arc (0:90:5pt) -- (66.4222pt,222.6738pt) -- (52.8222pt,222.6738pt) arc (90:180:5pt) -- (47.8222pt,217.6738pt) -- cycle;
\node at (54.8222pt,197.9277pt)[below right,color=black]{B};
\draw[pstyle1] (92.6483pt,195.9277pt) arc (180:270:5pt) -- (97.6483pt,190.9277pt) -- (111.5594pt,190.9277pt) arc (270:360:5pt) -- (116.5594pt,195.9277pt) -- (116.5594pt,217.6738pt) arc (0:90:5pt) -- (111.5594pt,222.6738pt) -- (97.6483pt,222.6738pt) arc (90:180:5pt) -- (92.6483pt,217.6738pt) -- cycle;
\node at (99.6483pt,197.9277pt)[below right,color=black]{C};
\draw[pstyle2] (57.6222pt,70.2246pt) -- (47.6222pt,66.2246pt);
\draw[pstyle2] (57.6222pt,70.2246pt) -- (47.6222pt,74.2246pt);
\draw[pstyle2] (17.6222pt,70.2246pt) -- (58.6222pt,70.2246pt);
\node at (24.6222pt,51.7461pt)[below right,color=black]{m1};
\draw[pstyle2] (90.6483pt,100.7031pt) -- (80.6483pt,96.7031pt);
\draw[pstyle2] (90.6483pt,100.7031pt) -- (80.6483pt,104.7031pt);
\draw[pstyle2] (17.6222pt,100.7031pt) -- (91.6483pt,100.7031pt);
\node at (24.6222pt,82.2246pt)[below right,color=black]{\guillemotleft create\guillemotright };
\draw[pstyle1] (92.6483pt,83.2246pt) arc (180:270:5pt) -- (97.6483pt,78.2246pt) -- (111.5594pt,78.2246pt) arc (270:360:5pt) -- (116.5594pt,83.2246pt) -- (116.5594pt,104.9707pt) arc (0:90:5pt) -- (111.5594pt,109.9707pt) -- (97.6483pt,109.9707pt) arc (90:180:5pt) -- (92.6483pt,104.9707pt) -- cycle;
\node at (99.6483pt,85.2246pt)[below right,color=black]{C};
\draw[pstyle2] (102.6039pt,143.4492pt) -- (92.6039pt,139.4492pt);
\draw[pstyle2] (102.6039pt,143.4492pt) -- (92.6039pt,147.4492pt);
\draw[pstyle2] (17.6222pt,143.4492pt) -- (103.6039pt,143.4492pt);
\node at (24.6222pt,124.9707pt)[below right,color=black]{m2};
\draw[pstyle2] (102.6039pt,173.9277pt) -- (92.6039pt,169.9277pt);
\draw[pstyle2] (102.6039pt,173.9277pt) -- (92.6039pt,177.9277pt);
\draw[pstyle2] (17.6222pt,173.9277pt) -- (103.6039pt,173.9277pt);
\node at (24.6222pt,155.4492pt)[below right,color=black]{\guillemotleft destroy\guillemotright };
\draw[pstyle3] (95.6039pt,164.9277pt) -- (113.6039pt,182.9277pt);
\draw[pstyle3] (95.6039pt,182.9277pt) -- (113.6039pt,164.9277pt);
\end{tikzpicture}
    \end{center}
    \caption{Creating and destroying lifelines. }
    \label{fig:asyncreate}
\end{figure}

Creation and deletion are deceptively simple. However, a closer inspection reveals subtle interactions between creation, deletion and combined fragments. For example, consider a diagram that describes a common worker thread scenario. A client schedules a job, the server creates a new process to run the job, and the process repeats. See Fig.~\ref{fig:serverbad}. 
What shall we do with the process's lifeline inside the loop fragment? If we destroy it, we kill the process each time we go through the loop. If not, then it is not clear what this lifeline represents beyond the loop's fragment. 

\begin{figure}[tb]
    \begin{center}
\definecolor{plantucolor0000}{RGB}{0,0,0}
\definecolor{plantucolor0001}{RGB}{24,24,24}
\definecolor{plantucolor0002}{RGB}{255,255,255}
\definecolor{plantucolor0003}{RGB}{238,238,238}
\begin{tikzpicture}[yscale=-1
,pstyle0/.style={color=black,line width=1.5pt}
,pstyle1/.style={color=plantucolor0001,line width=0.5pt,dash pattern=on 5.0pt off 5.0pt}
,pstyle2/.style={color=plantucolor0002,line width=0.5pt}
,pstyle4/.style={color=plantucolor0001,line width=1.0pt}
]
\draw[pstyle0] (10pt,54.7461pt) rectangle (239.3581pt,176.9277pt);
\draw[pstyle1] (31pt,37.7461pt) -- (31pt,224.4063pt);
\draw[pstyle1] (131.182pt,37.7461pt) -- (131.182pt,224.4063pt);
\draw[pstyle1] (217.7581pt,135.0762pt) -- (217.7581pt,224.4063pt);
\draw[pstyle2] (20pt,10pt) arc (180:270:5pt) -- (25pt,5pt) -- (38.9111pt,5pt) arc (270:360:5pt) -- (43.9111pt,10pt) -- (43.9111pt,31.7461pt) arc (0:90:5pt) -- (38.9111pt,36.7461pt) -- (25pt,36.7461pt) arc (90:180:5pt) -- (20pt,31.7461pt) -- cycle;
\node at (27pt,12pt)[below right,color=black]{C};
\draw[pstyle2] (20pt,228.4063pt) arc (180:270:5pt) -- (25pt,223.4063pt) -- (38.9111pt,223.4063pt) arc (270:360:5pt) -- (43.9111pt,228.4063pt) -- (43.9111pt,250.1523pt) arc (0:90:5pt) -- (38.9111pt,255.1523pt) -- (25pt,255.1523pt) arc (90:180:5pt) -- (20pt,250.1523pt) -- cycle;
\node at (27pt,230.4063pt)[below right,color=black]{C};
\draw[pstyle2] (120.182pt,10pt) arc (180:270:5pt) -- (125.182pt,5pt) -- (138.282pt,5pt) arc (270:360:5pt) -- (143.282pt,10pt) -- (143.282pt,31.7461pt) arc (0:90:5pt) -- (138.282pt,36.7461pt) -- (125.182pt,36.7461pt) arc (90:180:5pt) -- (120.182pt,31.7461pt) -- cycle;
\node at (127.182pt,12pt)[below right,color=black]{S};
\draw[pstyle2] (120.182pt,228.4063pt) arc (180:270:5pt) -- (125.182pt,223.4063pt) -- (138.282pt,223.4063pt) arc (270:360:5pt) -- (143.282pt,228.4063pt) -- (143.282pt,250.1523pt) arc (0:90:5pt) -- (138.282pt,255.1523pt) -- (125.182pt,255.1523pt) arc (90:180:5pt) -- (120.182pt,250.1523pt) -- cycle;
\node at (127.182pt,230.4063pt)[below right,color=black]{S};
\draw[pstyle2] (206.7581pt,228.4063pt) arc (180:270:5pt) -- (211.7581pt,223.4063pt) -- (224.3581pt,223.4063pt) arc (270:360:5pt) -- (229.3581pt,228.4063pt) -- (229.3581pt,250.1523pt) arc (0:90:5pt) -- (224.3581pt,255.1523pt) -- (211.7581pt,255.1523pt) arc (90:180:5pt) -- (206.7581pt,250.1523pt) -- cycle;
\node at (213.7581pt,230.4063pt)[below right,color=black]{P};
\draw[color=black,fill=plantucolor0003,line width=1.5pt] (10pt,54.7461pt) -- (81.4pt,54.7461pt) -- (81.4pt,63.2246pt) -- (71.4pt,73.2246pt) -- (10pt,73.2246pt) -- (10pt,54.7461pt);
\draw[pstyle0] (10pt,54.7461pt) rectangle (239.3581pt,176.9277pt);
\node at (25pt,55.7461pt)[below right,color=black]{\textbf{loop}};
\draw[pstyle4] (129.732pt,95.7031pt) -- (119.732pt,91.7031pt);
\draw[pstyle4] (129.732pt,95.7031pt) -- (119.732pt,99.7031pt);
\draw[pstyle4] (31.9556pt,95.7031pt) -- (130.732pt,95.7031pt);
\node at (38.9556pt,77.2246pt)[below right,color=black]{schedule job};
\draw[pstyle4] (204.7581pt,126.1816pt) -- (194.7581pt,122.1816pt);
\draw[pstyle4] (204.7581pt,126.1816pt) -- (194.7581pt,130.1816pt);
\draw[pstyle4] (131.732pt,126.1816pt) -- (205.7581pt,126.1816pt);
\node at (138.732pt,107.7031pt)[below right,color=black]{\guillemotleft create\guillemotright };
\draw[pstyle2] (206.7581pt,108.7031pt) arc (180:270:5pt) -- (211.7581pt,103.7031pt) -- (224.3581pt,103.7031pt) arc (270:360:5pt) -- (229.3581pt,108.7031pt) -- (229.3581pt,130.4492pt) arc (0:90:5pt) -- (224.3581pt,135.4492pt) -- (211.7581pt,135.4492pt) arc (90:180:5pt) -- (206.7581pt,130.4492pt) -- cycle;
\node at (213.7581pt,110.7031pt)[below right,color=black]{P};
\draw[pstyle4] (216.0581pt,168.9277pt) -- (206.0581pt,164.9277pt);
\draw[pstyle4] (216.0581pt,168.9277pt) -- (206.0581pt,172.9277pt);
\draw[pstyle4] (131.732pt,168.9277pt) -- (217.0581pt,168.9277pt);
\node at (138.732pt,150.4492pt)[below right,color=black]{run job};
\draw[pstyle4] (216.0581pt,206.4063pt) -- (206.0581pt,202.4063pt);
\draw[pstyle4] (216.0581pt,206.4063pt) -- (206.0581pt,210.4063pt);
\draw[pstyle4] (131.732pt,206.4063pt) -- (217.0581pt,206.4063pt);
\node at (138.732pt,187.9277pt)[below right,color=black]{?};
\end{tikzpicture}
    \end{center}
    \caption{We run into trouble when the lifeline models a process. The loop creates a new process in each iteration, but the diagram shows a single lifeline $P$. Does $P$ refer to a single process? which one?}
    \label{fig:serverbad}
\end{figure}
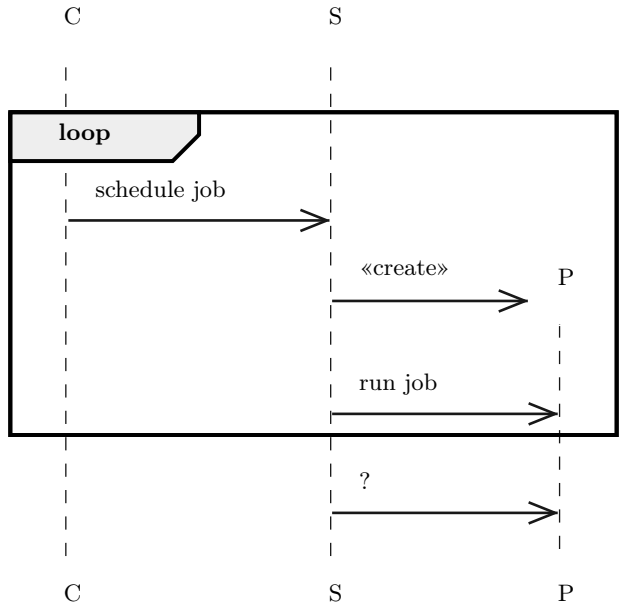

A way out of this problem is to define lifelines, not as direct representations of processes, but as names of variables that reference processes, and treat the body of combined fragments as a variables scope. 
Now, a $create$ message creates two things: a process and a local variable that refers to that process. The lifeline describes the local variable. Thus, it may reference in each iteration a different process, and when the loop terminates it goes out of scope, but the processes may continue. 

Another benefit of this interpretation is that we can use the $destroy$ message as an explicit request to terminate the process.  Thus, it is possible to destroy a lifeline, indicating that the variable has gone out of scope, without terminating the process. 
Or, to terminate the process but still send messages to the variable, for example, to query the process' termination status.

\begin{figure}
    \begin{center}
\definecolor{plantucolor0000}{RGB}{0,0,0}
\definecolor{plantucolor0001}{RGB}{24,24,24}
\definecolor{plantucolor0002}{RGB}{255,255,255}
\definecolor{plantucolor0003}{RGB}{238,238,238}
\begin{tikzpicture}[yscale=-1
,pstyle0/.style={color=black,line width=1.5pt}
,pstyle1/.style={color=plantucolor0001,line width=0.5pt,dash pattern=on 5.0pt off 5.0pt}
,pstyle2/.style={color=plantucolor0002,line width=0.5pt}
,pstyle4/.style={color=plantucolor0001,line width=1.0pt}
]
\draw[pstyle0] (10pt,54.7461pt) rectangle (239.3581pt,176.9277pt);
\draw[pstyle1] (31pt,37.7461pt) -- (31pt,224.4063pt);
\draw[pstyle1] (131.182pt,37.7461pt) -- (131.182pt,224.4063pt);
\draw[pstyle1] (217.7581pt,135.0762pt) -- (217.7581pt,176pt);
\draw[pstyle2] (20pt,10pt) arc (180:270:5pt) -- (25pt,5pt) -- (38.9111pt,5pt) arc (270:360:5pt) -- (43.9111pt,10pt) -- (43.9111pt,31.7461pt) arc (0:90:5pt) -- (38.9111pt,36.7461pt) -- (25pt,36.7461pt) arc (90:180:5pt) -- (20pt,31.7461pt) -- cycle;
\node at (27pt,12pt)[below right,color=black]{C};
\draw[pstyle2] (20pt,228.4063pt) arc (180:270:5pt) -- (25pt,223.4063pt) -- (38.9111pt,223.4063pt) arc (270:360:5pt) -- (43.9111pt,228.4063pt) -- (43.9111pt,250.1523pt) arc (0:90:5pt) -- (38.9111pt,255.1523pt) -- (25pt,255.1523pt) arc (90:180:5pt) -- (20pt,250.1523pt) -- cycle;
\node at (27pt,230.4063pt)[below right,color=black]{C};
\draw[pstyle2] (120.182pt,10pt) arc (180:270:5pt) -- (125.182pt,5pt) -- (138.282pt,5pt) arc (270:360:5pt) -- (143.282pt,10pt) -- (143.282pt,31.7461pt) arc (0:90:5pt) -- (138.282pt,36.7461pt) -- (125.182pt,36.7461pt) arc (90:180:5pt) -- (120.182pt,31.7461pt) -- cycle;
\node at (127.182pt,12pt)[below right,color=black]{S};
\draw[pstyle2] (120.182pt,228.4063pt) arc (180:270:5pt) -- (125.182pt,223.4063pt) -- (138.282pt,223.4063pt) arc (270:360:5pt) -- (143.282pt,228.4063pt) -- (143.282pt,250.1523pt) arc (0:90:5pt) -- (138.282pt,255.1523pt) -- (125.182pt,255.1523pt) arc (90:180:5pt) -- (120.182pt,250.1523pt) -- cycle;
\node at (127.182pt,230.4063pt)[below right,color=black]{S};
\draw[pstyle2] (206.7581pt,228.4063pt) arc (180:270:5pt) -- (211.7581pt,223.4063pt) -- (224.3581pt,223.4063pt) arc (270:360:5pt) -- (229.3581pt,228.4063pt) -- (229.3581pt,250.1523pt) arc (0:90:5pt) -- (224.3581pt,255.1523pt) -- (211.7581pt,255.1523pt) arc (90:180:5pt) -- (206.7581pt,250.1523pt) -- cycle;
\draw[color=black,fill=plantucolor0003,line width=1.5pt] (10pt,54.7461pt) -- (81.4pt,54.7461pt) -- (81.4pt,63.2246pt) -- (71.4pt,73.2246pt) -- (10pt,73.2246pt) -- (10pt,54.7461pt);
\draw[pstyle0] (10pt,54.7461pt) rectangle (239.3581pt,176.9277pt);
\node at (25pt,55.7461pt)[below right,color=black]{\textbf{loop}};
\draw[pstyle4] (129.732pt,95.7031pt) -- (119.732pt,91.7031pt);
\draw[pstyle4] (129.732pt,95.7031pt) -- (119.732pt,99.7031pt);
\draw[pstyle4] (31.9556pt,95.7031pt) -- (130.732pt,95.7031pt);
\node at (38.9556pt,77.2246pt)[below right,color=black]{schedule job};
\draw[pstyle4] (204.7581pt,126.1816pt) -- (194.7581pt,122.1816pt);
\draw[pstyle4] (204.7581pt,126.1816pt) -- (194.7581pt,130.1816pt);
\draw[pstyle4] (131.732pt,126.1816pt) -- (205.7581pt,126.1816pt);
\node at (138.732pt,107.7031pt)[below right,color=black]{\guillemotleft create\guillemotright };
\draw[pstyle2] (206.7581pt,108.7031pt) arc (180:270:5pt) -- (211.7581pt,103.7031pt) -- (224.3581pt,103.7031pt) arc (270:360:5pt) -- (229.3581pt,108.7031pt) -- (229.3581pt,130.4492pt) arc (0:90:5pt) -- (224.3581pt,135.4492pt) -- (211.7581pt,135.4492pt) arc (90:180:5pt) -- (206.7581pt,130.4492pt) -- cycle;
\node at (213.7581pt,110.7031pt)[below right,color=black]{P};
\draw[pstyle4] (216.0581pt,168.9277pt) -- (206.0581pt,164.9277pt);
\draw[pstyle4] (216.0581pt,168.9277pt) -- (206.0581pt,172.9277pt);
\draw[pstyle4] (131.732pt,168.9277pt) -- (217.0581pt,168.9277pt);
\node at (138.732pt,150.4492pt)[below right,color=black]{run job};
\end{tikzpicture}
    \end{center}
    \caption{The lifeline models a variable and the loop's body is its scope. The $create$ message creates a new process and a local variable $P$ to refer to that process. Each iteration uses the same variable to reference a newly created process. The variable $P$ (and hence the lifeline it denotes) is local to the loop's body. }
    \label{fig:servergood}
\end{figure}
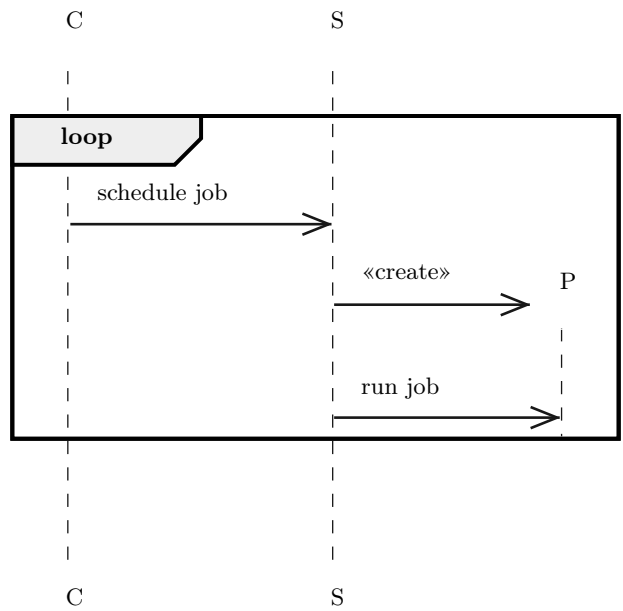

The semantics of a $create$ or a $destroy$ message does not affect the set of traces, it only affects the set of lifeline names that are valid immediately after these messages occur:
\begin{zed}
    N(create(n), ns) =  ns \cup \{ n \} ) \also
    N(destroy(n), ns) = ns \setminus \{ n\} \also
    D(create(n), ns) = \{ \nil \} \also
    D(destroy(n), ns) = \{ \nil \} 
\end{zed}

This definition makes it explicit that creating or destroying a lifeline is not a communication event, it only affects the namespace in which we calculate the rest of the trace.

\section{Related work}
\label{related}

As there is a lot of work on the semantics of sequence diagrams, we will not delve into the details of each work.
Instead, we will outline the major directions, indicating where our work differs and highlighting our contribution. 
A detailed review of existing work is available in~\cite{MAHE24}.

There are two major approaches to defining the semantics of sequence diagrams. The first approach, motivated mainly by the desire to apply formal verification techniques, is to translate the diagrams into an existing formal notation, for example, into a process algebra. This immediately provides a well-defined semantics. However, it is often difficult to translate all the features (for example, nested fragments are difficult to model in CSP). As a result, such work often leaves out important features. 
In addition, explaining the semantics to engineers requires them to master the formal notation and the translation process. 
This makes it more difficult to explain and hinders the acceptance of the semantics for many practitioners.

Many earlier works belonging to this category are reviewed in~\cite{MicskeiW11}. 
Later works along these lines are, for example, Jacobs and Simpson~\cite{jacobs2015process} that use CSP as the underlying semantics. However, they do not cover object creation/destruction or fragment nesting. 

In~\cite{LimaIS14} the authors translate sequence diagrams into the COMPASS modeling language. However, the translation is defined informally, there is no explicit semantics, and the paper does not consider weak sequential composition or object creation/deletion. 

The authors of~\cite{ChenML20}, translate sequence diagrams into Clock Constraint Specification Language (CCSL) constraints.
The translation is given in terms of examples, and sequence diagrams are not given explicit semantics, making it difficult to assess the meaning of the results. 

In~\cite{MouakherDA22} sequence diagrams are translated to an Event-B machine. Its drawbacks are that there is no clear mapping between the sequence diagram notation and the formal structure, and the semantics itself is much more complicated than a simple set of traces. This makes it difficult to understand and creates a higher barrier to adoption by practitioners.

Finally,~\cite{HamroucheCM22} translates sequence diagrams using graph transformations to CSP. The CSP code is analyzed using the FDR4 model checker. However, the paper does not provide explicit semantics, instead, it describes a large (but partial) set of graph transformation rules that translate from one syntax to another. 
It is difficult to understand what is the semantics that the authors had in mind or ascertain if the transformation matches the desired semantics.

In summary, the first approach is based on a translation process that generates a structure in an existing formal notation. 
The semantics of sequence diagrams is thus implicitly given as the semantics of the generated structure. 
We agree with~\cite{MicskeiW11} that while translation to an existing formal notation benefits from the tools that support the notation, it makes the approach cumbersome in practice because the translation is complicated and requires mastery of mathematical concepts that are daunting to most engineers.


The second approach is to define the semantics directly on the sequence diagram notation. This approach is appealing because it maps directly to the structure of the diagrams and is based on a set of traces, a concept that is easy to understand and is closer to the informal semantics described in the UML standard. 

For example,~\cite{LuK14} defines an abstract syntax for sequence diagrams with trace semantics that associates each AST node with a set of traces. The work focuses on defining refinement relations between sequence diagrams. It does not address the difference between synchronous and asynchronous messages, but the semantics appear to be suitable only for asynchronous messages, as there are no rules to enforce the restrictions of synchronous messages. 
In addition, the work does not consider the creation or destruction of lifelines. 
Finally, their semantic rules are much more complicated compared to our work. 

The authors of~\cite{DhaouMAB17} extend previous work that focuses on defining a partial order between events in sequence diagrams. However, they do not provide explicit rules for calculating the traces of a sequence diagram. 

In~\cite{HagaMC21} the authors use an abstract syntax tree to define the structure of sequence diagrams, and a state machine formalism to specify their operational semantics. Their semantics considers weak composition only for a sequence of primitive messages. Their work does not provide a denotational semantics. Nor do they treat object creation/destruction. 

Finally, in~\cite{MAHE24} the authors provide both a denotational and an operational semantics for sequence diagrams, and prove their equivalence. Compared to this work, we provide a simpler semantics at the price of treating messages as single atomic events instead of two separate send and receive events. In addition, we define the semantics for creating and destroying lifelines, and illustrate how it interacts withe meaning of combined fragments, a topic that is missing from prior studies on the semantics of sequence diagrams. 

\section{Discussion}
\label{discusssion}

\subsection{Modeling messages}

The UML standard separates sending and receiving messages into different events. Accordingly, most works on formalizing the semantics treat sending and receiving messages as separate events. This makes it possible to model typical scenarios in distributed systems where messages are lost or received out of order. 
However, the price we pay for this flexibility is a more complicated semantics and a much larger set of possible traces. 
In our work, in contrast, we treat sending and receiving a message as a single atomic event.
It may seem that we simplify the semantics at the price of losing the ability to model important phenomena. 
But we can use a simple trick to model all that was apparently lost, by introducing a lifeline to represent the message medium. 
A lost message is a message that is sent to the medium but not forwarded to its destination.
A delay is achieved by emitting the message from the medium after the required delay duration. 
We may also use the medium to model duplicating messages or sending messages out of order. See Fig.~\ref{fig:medium}.

\begin{figure}
    \begin{center}
\definecolor{plantucolor0000}{RGB}{24,24,24}
\definecolor{plantucolor0001}{RGB}{255,255,255}
\definecolor{plantucolor0002}{RGB}{0,0,0}
\begin{tikzpicture}[yscale=-1
,pstyle0/.style={color=plantucolor0000,line width=0.5pt,dash pattern=on 5.0pt off 5.0pt}
,pstyle1/.style={color=plantucolor0001,line width=0.5pt}
,pstyle2/.style={color=plantucolor0000,line width=1.0pt}
]
\draw[pstyle0] (17pt,37.7461pt) -- (17pt,210.1387pt);
\draw[pstyle0] (88.77pt,37.7461pt) -- (88.77pt,210.1387pt);
\draw[pstyle0] (176.7677pt,37.7461pt) -- (176.7677pt,210.1387pt);
\draw[pstyle1] (5pt,10pt) arc (180:270:5pt) -- (10pt,5pt) -- (25.2444pt,5pt) arc (270:360:5pt) -- (30.2444pt,10pt) -- (30.2444pt,31.7461pt) arc (0:90:5pt) -- (25.2444pt,36.7461pt) -- (10pt,36.7461pt) arc (90:180:5pt) -- (5pt,31.7461pt) -- cycle;
\node at (12pt,12pt)[below right,color=black]{A};
\draw[pstyle1] (5pt,214.1387pt) arc (180:270:5pt) -- (10pt,209.1387pt) -- (25.2444pt,209.1387pt) arc (270:360:5pt) -- (30.2444pt,214.1387pt) -- (30.2444pt,235.8848pt) arc (0:90:5pt) -- (25.2444pt,240.8848pt) -- (10pt,240.8848pt) arc (90:180:5pt) -- (5pt,235.8848pt) -- cycle;
\node at (12pt,216.1387pt)[below right,color=black]{A};
\draw[pstyle1] (54.77pt,10pt) arc (180:270:5pt) -- (59.77pt,5pt) -- (118.2744pt,5pt) arc (270:360:5pt) -- (123.2744pt,10pt) -- (123.2744pt,31.7461pt) arc (0:90:5pt) -- (118.2744pt,36.7461pt) -- (59.77pt,36.7461pt) arc (90:180:5pt) -- (54.77pt,31.7461pt) -- cycle;
\node at (61.77pt,12pt)[below right,color=black]{Medium};
\draw[pstyle1] (54.77pt,214.1387pt) arc (180:270:5pt) -- (59.77pt,209.1387pt) -- (118.2744pt,209.1387pt) arc (270:360:5pt) -- (123.2744pt,214.1387pt) -- (123.2744pt,235.8848pt) arc (0:90:5pt) -- (118.2744pt,240.8848pt) -- (59.77pt,240.8848pt) arc (90:180:5pt) -- (54.77pt,235.8848pt) -- cycle;
\node at (61.77pt,216.1387pt)[below right,color=black]{Medium};
\draw[pstyle1] (165.7677pt,10pt) arc (180:270:5pt) -- (170.7677pt,5pt) -- (184.3677pt,5pt) arc (270:360:5pt) -- (189.3677pt,10pt) -- (189.3677pt,31.7461pt) arc (0:90:5pt) -- (184.3677pt,36.7461pt) -- (170.7677pt,36.7461pt) arc (90:180:5pt) -- (165.7677pt,31.7461pt) -- cycle;
\node at (172.7677pt,12pt)[below right,color=black]{B};
\draw[pstyle1] (165.7677pt,214.1387pt) arc (180:270:5pt) -- (170.7677pt,209.1387pt) -- (184.3677pt,209.1387pt) arc (270:360:5pt) -- (189.3677pt,214.1387pt) -- (189.3677pt,235.8848pt) arc (0:90:5pt) -- (184.3677pt,240.8848pt) -- (170.7677pt,240.8848pt) arc (90:180:5pt) -- (165.7677pt,235.8848pt) -- cycle;
\node at (172.7677pt,216.1387pt)[below right,color=black]{B};
\draw[pstyle2] (87.0222pt,70.2246pt) -- (77.0222pt,66.2246pt);
\draw[pstyle2] (87.0222pt,70.2246pt) -- (77.0222pt,74.2246pt);
\draw[pstyle2] (17.6222pt,70.2246pt) -- (88.0222pt,70.2246pt);
\node at (24.6222pt,51.7461pt)[below right,color=black]{m1 to B};
\draw[pstyle2] (87.0222pt,100.7031pt) -- (77.0222pt,96.7031pt);
\draw[pstyle2] (87.0222pt,100.7031pt) -- (77.0222pt,104.7031pt);
\draw[pstyle2] (17.6222pt,100.7031pt) -- (88.0222pt,100.7031pt);
\node at (24.6222pt,82.2246pt)[below right,color=black]{m2 to B};
\draw[pstyle2] (87.0222pt,131.1816pt) -- (77.0222pt,127.1816pt);
\draw[pstyle2] (87.0222pt,131.1816pt) -- (77.0222pt,135.1816pt);
\draw[pstyle2] (17.6222pt,131.1816pt) -- (88.0222pt,131.1816pt);
\node at (24.6222pt,112.7031pt)[below right,color=black]{m3 to B};
\draw[pstyle2] (175.5677pt,161.6602pt) -- (165.5677pt,157.6602pt);
\draw[pstyle2] (175.5677pt,161.6602pt) -- (165.5677pt,165.6602pt);
\draw[pstyle2] (89.0222pt,161.6602pt) -- (176.5677pt,161.6602pt);
\node at (96.0222pt,143.1816pt)[below right,color=black]{m3 from A};
\draw[pstyle2] (175.5677pt,192.1387pt) -- (165.5677pt,188.1387pt);
\draw[pstyle2] (175.5677pt,192.1387pt) -- (165.5677pt,196.1387pt);
\draw[pstyle2] (89.0222pt,192.1387pt) -- (176.5677pt,192.1387pt);
\node at (96.0222pt,173.6602pt)[below right,color=black]{m2 from A};
\end{tikzpicture}
    \end{center}
    \caption{A medium lifeline may be used to model messages that get lost or arrive out of order. Here $A$ sends three messages to $B$. The first one is lost and the other two arrive out of order.}
    \label{fig:medium}
\end{figure}
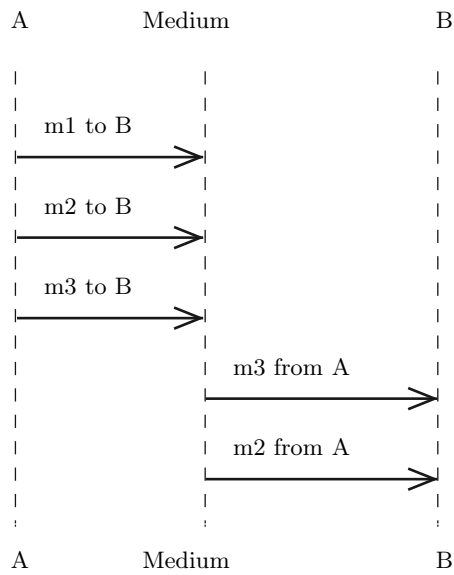

\subsection{Synchronous interactions}

Synchronous interactions place several restrictions on the patterns of message exchange. First, every synchronous call must be accompanied by a reply message. Second, the structure of the calls must honor the stack-based activation semantics. These restrictions complicate the semantics, in particular in the case of the parallel fragment. 

We may model synchronous diagrams using our asynchronous semantics by excluding any traces that violate these restrictions. However, such restrictions interfere with our desire to keep the semantics compositional. Thus, a proper development of these ideas is too large to fit in this work. 

\subsection{Excluding guards and conditions}

We have decided to exclude guard and condition expressions from alts and loops. This makes the semantics much easier to describe, analyze, reason about and in general understand, and this outweighs the disadvantages of not representing them in the semantics because it provides a sound understanding of what the diagrams mean, and it even makes it easier to understand how the guards affect the meaning of the unguarded operators. 
In particular, the guards act as filters that restrict the set of traces, thus the unguarded versions are an over-approximation of the guarded versions. When reasoning informally we can understand what traces they exclude. When reasoning formally, for example when we use a tool, we may implement a more elaborate semantics that includes support for such expressions.

We may make an analogy with the difference between the theory of regular expressions and their use in practice, which adds many kinds of bells and whistles that are important for practical use but would considerably complicate the formal semantics if introduced into the theory (for example anchors, named subpatterns, and even context aware constructs like backreferences, that in effect make the language more powerful than a regular language). 

\subsection{Messages to self}
\label{self}

The standard syntax to represent a message from a lifeline to itself is shown in Fig.~\ref{fig:asyncself}. 
But there is a problem with this notation. We use horizontal lines to draw messages between different lifelines. The horizontal line means that sending a message is an instantaneous event. 
Yet, for self-messages, the line is not horizontal. Does it mean that self-messages are not instantaneous? A more consistent notation is to draw self-messages as dots on the horizontal line. See Fig.~\ref{fig:asyncself2}.

\begin{figure}
    \begin{center}
\definecolor{plantucolor0000}{RGB}{24,24,24}
\definecolor{plantucolor0001}{RGB}{255,255,255}
\definecolor{plantucolor0002}{RGB}{0,0,0}
\begin{tikzpicture}[yscale=-1
,pstyle0/.style={color=plantucolor0000,line width=0.5pt,dash pattern=on 5.0pt off 5.0pt}
,pstyle1/.style={color=plantucolor0001,line width=0.5pt}
,pstyle2/.style={color=plantucolor0000,line width=1.0pt}
]
\draw[pstyle0] (17pt,37.7461pt) -- (17pt,131.7031pt);
\draw[pstyle0] (61.8222pt,37.7461pt) -- (61.8222pt,131.7031pt);
\draw[pstyle1] (5pt,10pt) arc (180:270:5pt) -- (10pt,5pt) -- (25.2444pt,5pt) arc (270:360:5pt) -- (30.2444pt,10pt) -- (30.2444pt,31.7461pt) arc (0:90:5pt) -- (25.2444pt,36.7461pt) -- (10pt,36.7461pt) arc (90:180:5pt) -- (5pt,31.7461pt) -- cycle;
\node at (12pt,12pt)[below right,color=black]{A};
\draw[pstyle1] (5pt,135.7031pt) arc (180:270:5pt) -- (10pt,130.7031pt) -- (25.2444pt,130.7031pt) arc (270:360:5pt) -- (30.2444pt,135.7031pt) -- (30.2444pt,157.4492pt) arc (0:90:5pt) -- (25.2444pt,162.4492pt) -- (10pt,162.4492pt) arc (90:180:5pt) -- (5pt,157.4492pt) -- cycle;
\node at (12pt,137.7031pt)[below right,color=black]{A};
\draw[pstyle1] (50.8222pt,10pt) arc (180:270:5pt) -- (55.8222pt,5pt) -- (69.4222pt,5pt) arc (270:360:5pt) -- (74.4222pt,10pt) -- (74.4222pt,31.7461pt) arc (0:90:5pt) -- (69.4222pt,36.7461pt) -- (55.8222pt,36.7461pt) arc (90:180:5pt) -- (50.8222pt,31.7461pt) -- cycle;
\node at (57.8222pt,12pt)[below right,color=black]{B};
\draw[pstyle1] (50.8222pt,135.7031pt) arc (180:270:5pt) -- (55.8222pt,130.7031pt) -- (69.4222pt,130.7031pt) arc (270:360:5pt) -- (74.4222pt,135.7031pt) -- (74.4222pt,157.4492pt) arc (0:90:5pt) -- (69.4222pt,162.4492pt) -- (55.8222pt,162.4492pt) arc (90:180:5pt) -- (50.8222pt,157.4492pt) -- cycle;
\node at (57.8222pt,137.7031pt)[below right,color=black]{B};
\draw[pstyle2] (60.6222pt,70.2246pt) -- (50.6222pt,66.2246pt);
\draw[pstyle2] (60.6222pt,70.2246pt) -- (50.6222pt,74.2246pt);
\draw[pstyle2] (17.6222pt,70.2246pt) -- (61.6222pt,70.2246pt);
\node at (24.6222pt,51.7461pt)[below right,color=black]{m1};
\draw[pstyle2] (17.6222pt,100.7031pt) -- (59.6222pt,100.7031pt);
\draw[pstyle2] (59.6222pt,100.7031pt) -- (59.6222pt,113.7031pt);
\draw[pstyle2] (18.6222pt,113.7031pt) -- (59.6222pt,113.7031pt);
\draw[pstyle2] (18.6222pt,113.7031pt) -- (28.6222pt,109.7031pt);
\draw[pstyle2] (18.6222pt,113.7031pt) -- (28.6222pt,117.7031pt);
\node at (24.6222pt,82.2246pt)[below right,color=black]{m2};
\end{tikzpicture}
    \end{center}
    \caption{$A$ sends message $m_2$ to itself. The common notation to indicate messages sent to self conveys the false impression that sending a message to self is not instantaneous.}
    \label{fig:asyncself}
\end{figure}
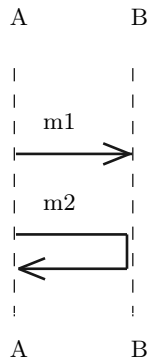

\begin{figure}
    \begin{center}
\definecolor{plantucolor0000}{RGB}{24,24,24}
\definecolor{plantucolor0001}{RGB}{255,255,255}
\definecolor{plantucolor0002}{RGB}{0,0,0}
\begin{tikzpicture}[yscale=-1
,pstyle0/.style={color=plantucolor0000,line width=0.5pt,dash pattern=on 5.0pt off 5.0pt}
,pstyle1/.style={color=plantucolor0001,line width=0.5pt}
,pstyle2/.style={color=plantucolor0000,line width=1.0pt}
]
\draw[pstyle0] (17pt,37.7461pt) -- (17pt,131.7031pt);
\draw[pstyle0] (61.8222pt,37.7461pt) -- (61.8222pt,131.7031pt);
\draw[pstyle1] (5pt,10pt) arc (180:270:5pt) -- (10pt,5pt) -- (25.2444pt,5pt) arc (270:360:5pt) -- (30.2444pt,10pt) -- (30.2444pt,31.7461pt) arc (0:90:5pt) -- (25.2444pt,36.7461pt) -- (10pt,36.7461pt) arc (90:180:5pt) -- (5pt,31.7461pt) -- cycle;
\node at (12pt,12pt)[below right,color=black]{A};
\draw[pstyle1] (5pt,135.7031pt) arc (180:270:5pt) -- (10pt,130.7031pt) -- (25.2444pt,130.7031pt) arc (270:360:5pt) -- (30.2444pt,135.7031pt) -- (30.2444pt,157.4492pt) arc (0:90:5pt) -- (25.2444pt,162.4492pt) -- (10pt,162.4492pt) arc (90:180:5pt) -- (5pt,157.4492pt) -- cycle;
\node at (12pt,137.7031pt)[below right,color=black]{A};
\draw[pstyle1] (50.8222pt,10pt) arc (180:270:5pt) -- (55.8222pt,5pt) -- (69.4222pt,5pt) arc (270:360:5pt) -- (74.4222pt,10pt) -- (74.4222pt,31.7461pt) arc (0:90:5pt) -- (69.4222pt,36.7461pt) -- (55.8222pt,36.7461pt) arc (90:180:5pt) -- (50.8222pt,31.7461pt) -- cycle;
\node at (57.8222pt,12pt)[below right,color=black]{B};
\draw[pstyle1] (50.8222pt,135.7031pt) arc (180:270:5pt) -- (55.8222pt,130.7031pt) -- (69.4222pt,130.7031pt) arc (270:360:5pt) -- (74.4222pt,135.7031pt) -- (74.4222pt,157.4492pt) arc (0:90:5pt) -- (69.4222pt,162.4492pt) -- (55.8222pt,162.4492pt) arc (90:180:5pt) -- (50.8222pt,157.4492pt) -- cycle;
\node at (57.8222pt,137.7031pt)[below right,color=black]{B};
\draw[pstyle2] (60.6222pt,70.2246pt) -- (50.6222pt,66.2246pt);
\draw[pstyle2] (60.6222pt,70.2246pt) -- (50.6222pt,74.2246pt);
\draw[pstyle2] (17.6222pt,70.2246pt) -- (61.6222pt,70.2246pt);
\node at (24.6222pt,51.7461pt)[below right,color=black]{m1};

\draw[black,fill=black] (17.6222pt,100.7031pt) circle (.5ex);
\node at (24.6222pt,82.2246pt)[below right,color=black]{m2};
\end{tikzpicture}
    \end{center}
    \caption{$A$ sends message $m_2$ to itself. Using a dot to draw a message to self is consistent with an instantaneous interpretation.}
    \label{fig:asyncself2}
\end{figure}
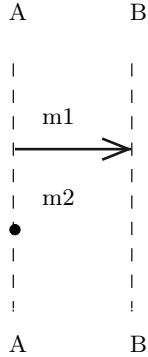

\subsection{Other combined fragments}

The UML provides additional combined fragments that are less widely used, and therefore we have not included them in the main body of our semantics. 
In this section we briefly discuss each one, either suggesting how it may be incorporated into the semantics or why we do not or cannot model it with our semantics. 

\subsubsection{Critical}

The $critical$ fragment protects the traces in its body from being interleaved. We can model this behavior by adding a transaction data structure to our messages. That is, a message is either a basic (single arrow) event or a transaction, which is a sequence of messages. The meaning of a critical fragment is then to wrap its messages inside a transaction. We did not add this to the semantics because we feel that the infrequent use of $critical$ does not justify this complication.

\subsubsection{Assert and Negate}

The stated purpose of $assert$ and $negate$ is to specify required and forbidden traces. 
However, as noted in~\cite{HarelM08}, this purpose cannot be achieved using a simple trace semantics. 
Indeed, the entire question of how to interpret the set of traces described by a sequence diagram with respect to a system, is beyond the scope of the diagram's semantics. 
Let us denote the set of traces described by the diagram as $D$ and the set of traces exhibited by a system as $S$. We may use a sequence diagram to specify required behaviors, in that case each trace described by the diagram must be a possible valid system behavior. 
That is, $D \subseteq S$. Alternatively, we may use the diagram as a specification of all possible behaviors. 
In that case, each system trace must be a member of the diagram's traces. That is $S \subseteq D$. 
Or we may use the diagram to specify forbidden behaviors, and in this case no trace of the diagram must be a possible system behavior. 
That is, $D \cap S = \emptyset$. The point of this discussion is that there is no need for (and no way to define) operators such as $assert$ and $negate$ because their intended meaning lies in the way we use the diagram in relation to a system, and not in how we determine the set of traces defined by the diagram. 

\subsubsection{Consider and Ignore}

Finally, the $consider$ and $ignore$ fragments explicitly filter the set of traces by specifying which messages to include or exclude from the set of traces defined by their bodies.
Thus, to add them to our semantics we could define

\begin{zed}
    D(consider(ms, b)) = \{ t : D(b) @ filter(ms,t) \}
\also  
    D(ignore(ms, b)) = \{ t : D(b) @ filter(\Sigma\setminus ms,t) \}
\end{zed}
where $filter(ms,t)$ removes from a trace $t$ all messages that are not members of the set $ms$, and $\Sigma$ is the set of all messages.

\section{Summary and future work}
\label{summary}
We have described a simple denotational trace semantics for the most commonly used fragments of asynchronous sequence diagrams. 
In contrast with other work, our semantics explicitly considers---and therefore clarifies, the creation and destruction of lifelines. We have shown that even though our design choices simplify the semantics by treating a message exchange as an atomic event, we do not lose expressive power. 

This work sets the basis for two important future developments. First, as we have hinted in the discussion, we plan to use the semantics as the basis for specifying the synchronous-diagrams semantics. Second, we intend to develop an operational semantics to accompany our denotational semantics and use it to provide tool support for reasoning about sequence diagrams. 
Finally, we plan to use the semantics that we have described in this paper to develop a precise notion of abstraction/refinement relations between sequence diagrams.

\bibliographystyle{plain}  
\bibliography{main}

\section*{Appendix}

\subsection{Summary of semantics}

We use two functions to define the semantics. The function $N$ calculates each fragment's namespace; the function $D$ calculates each fragment's traces, it uses $N$ to check that the peers of each message exist in the fragment's namespace.

\begin{zed}
    N(create(n), ns) =  ns \cup \{ n \} 
    \also
    N(destroy(n), ns) = ns \setminus \{ n\} 
    \also
    N(weakseq(x,y), ns) = N(y, N(x, ns))
    \also
    N(x, ns) = ns
 \end{zed}

\begin{zed}
    D(create(n), ns) =  \{ \nil \}
    \also
    D(destroy(n), ns) = \{ \nil \}
    \also
    D(basic(m_1,\ldots, m_n), ns) = weak~\trace{m_1,\ldots,m_n} \qquad \hbox{provided that $peers(m) \in ns$.}
    \also
    D(weakseq(x, y), ns) = weak~D(x, ns) \cat D(y, N(weakseq(x, y), ns))
    \also
    D(alt(x, y), ns) = D(x, ns) \cup D(y, ns)
    \also
    D(par(x, y), ns) = interleave~(D(x, ns), D(y, ns))
    \also
    D(loop(b), ns) = \bigcup \{ t : D(b, ns)^* @ weak(t) \}
\end{zed}

\subsection{Weak sequential composition}

The weak sequential composition function calculates a set of traces from a sequence of messages. 
It keeps the sequential order between messages that share a common lifeline. Messages that do not share a common lifeline interleave.

\begin{zed}
    weak~\nil = \nil \\
    weak~\trace{x} = \trace{x} \\
    weak~\trace{x,y}\cat t = x \cat weak~(\trace{y}\cat t) \qquad \hbox{if $x$ and $y$ share a common lifeline}\\
    weak~\trace{x,y}\cat t = x \cat weak~(\trace{y}\cat t) \cup y \cat weak~(\trace{x} \cat t) \qquad {\hbox{otherwise} }\\
\end{zed}

The concatenation operator is overloaded, when applied to a single message and a set of traces it concatenates the message to each trace in the set.

\subsection{Interleaving}

To interleave two sets of traces, we interleave each trace in the first set with every trace in the second set.

\begin{zed}
    interleave~xs~ys = \bigcup \{ x : xs , y : ys @ x \interleave y \}
\end{zed}

To interleave two traces, we embed one trace in the other while keeping their original order unchanged.

\begin{zed}
\nil \interleave t = t \\
t \interleave \nil = t \\
x \interleave y = \{ t : \Tail x \interleave y @ \trace{\Head x}\cat t \} \union \{ t : x \interleave \Tail y @ \trace{\Head y} \cat t \}
\end{zed}

\subsection{Theorems}

\begin{theorem}
$$D(loop(b)) = D(alt(skip, weakseq(b), weakseq(b,b), weakseq(b,b,b), \ldots))$$ 
\end{theorem}

\begin{proof}
We prove the theorem by expanding each side of the equation until we reach the same expression on both sides. 
Starting with the left-hand side we argue,

\begin{argue}
    D(loop(b)) \\
     = \bigcup\{ t : D(b)^* @ weak(t) \} & by definition of $D(loop)$ \\
     = \bigcup \{ t : \bigcup\limits_{n} D(b)^n @ weak(t) \} & by definition of Kleene closure \\
     = \bigcup\limits_{n} \{ t : D(b)^n @ weak(t) \} & by associativity \& commutativity of union
\end{argue}

Expanding the right-hand side, 
\begin{argue}
    D(alt(skip, weakseq(b), weakseq(b,b), weakseq(b,b,b), \ldots)) \\
    = \bigcup\limits_{n} \{ D(weakseq(\overbrace{b, \ldots, b}^{n}) \} & definition of $D(alt)$ \\
    = \bigcup\limits_{n} \{ \bigcup\limits \{ t : D(b)^n @ weak(t) \} \} & definition of $D(weakseq)$ \\
    = \bigcup\limits_{n} \{ t : D(b)^n @ weak(t) \} & by associativity \& commutativity of union
\end{argue}
As both sides reduce to the same expression, they are equal. \qed
\end{proof}
\end{document}